\documentclass[journal]{IEEEtran}

\usepackage{epsfig,latexsym,amsmath,amssymb,setspace,cite,color,graphicx,algorithm,algorithmic,cases}
\usepackage[caption=false,font=footnotesize]{subfig}
\usepackage[percent]{overpic}
\newcommand\numberthis{\addtocounter{equation}{1}\tag{\theequation}}

\usepackage{physics}
\usepackage{amsmath}
\usepackage{amsfonts}
\usepackage{amssymb}
\usepackage{dsfont}
\usepackage{bm}
\usepackage{amsthm}
\usepackage{newlfont}
\usepackage{float}
\usepackage{hyperref}
\usepackage{algorithm}
\usepackage{algorithmic}
\usepackage{enumerate}
\usepackage{chngcntr}
\usepackage{mathtools}
\usepackage{breqn}
\usepackage{stmaryrd}
\usepackage{cite}
\usepackage[yyyymmdd]{datetime}
\hypersetup{
    colorlinks=true, 
    linkcolor= blue, 
    citecolor=blue 
}

\DeclareMathOperator*{\argmin}{arg\,min}

\begin{document}
\sloppy
\allowdisplaybreaks[1]

\newtheorem{thm}{Theorem} 
\newtheorem{lem}{Lemma}
\newtheorem{prop}{Proposition}
\newtheorem{cor}{Corollary}
\newtheorem{defn}{Definition}
\newcommand{\remarkend}{\IEEEQEDopen}
\newtheorem{remark}{Remark}
\newtheorem{rem}{Remark}
\newtheorem{ex}{Example}
\newtheorem{pro}{Property}

\newenvironment{example}[1][Example]{\begin{trivlist}
\item[\hskip \labelsep {\bfseries #1}]}{\end{trivlist}}

\renewcommand{\qedsymbol}{ \begin{tiny}$\blacksquare$ \end{tiny} }

\renewcommand{\leq}{\leqslant}
\renewcommand{\geq}{\geqslant}

\title {Private Classical Communication over\\Quantum Multiple-Access Channels
}

\author{\IEEEauthorblockN{R\'emi A. Chou}\\

\thanks{
R. Chou is with the Department of Electrical
Engineering and Computer Science, Wichita State University, Wichita, KS. Part of this work has been presented at the 2021 IEEE International Symposium on Information Theory~(ISIT) \cite{chou2021private}. This work was supported in part by NSF grants CCF-1850227 and CCF-2047913. E-mail: remi.chou@wichita.edu. }
}
\maketitle
\begin{abstract}
We study private classical communication over quantum multiple-access channels. For an arbitrary number of transmitters, we derive a regularized expression of the capacity region. In the case of degradable channels, we establish a single-letter expression for the best achievable sum-rate and prove that this quantity also corresponds to the best achievable sum-rate for quantum communication over degradable quantum multiple-access channels. In our achievability result, we decouple the reliability and privacy constraints, which are handled via source coding with quantum side information and universal hashing, respectively. Hence, we also establish that the multi-user coding problem under consideration can be handled solely via point-to-point coding techniques. As a by-product of independent interest, we derive a distributed leftover hash lemma against quantum side information that ensures privacy in our achievability~result.
\end{abstract}

\section{Introduction}

The capacity of private classical communication over point-to-point quantum channels has been characterized in \cite{cai2004quantum,devetak2005private}. While only a regularized expression of this capacity is known,  a single-letter expression has been obtained in the case of degradable quantum channels \cite{smith2008private}, and coincides with the coherent information of the channel.  In this paper, we define private classical communication over quantum multiple-access  channels, and determine  a regularized expression of the capacity region for an arbitrary number of transmitters. As formally described in the next sections, we consider message indistinguishability as privacy metric. Our proposed setting can be seen as a quantum counterpart to the classical multiple-access~wiretap channel, first introduced in \cite{tekin2008general} and further studied in \cite{yassaee2010multiple,pierrot2011strongly,wiese2013strong,chen2018collective,hayashi2019secrecy}. Note that  for the special case of  classical communication over multiple-access quantum channels \emph{without privacy constraint}, the capacity region has already been characterized in~\cite{winter2001capacity}.

Often, for simplicity and to facilitate the design of good codes, coding for multiple-access channels is reduced to point-point coding techniques, for instance, with successive decoding or rate-splitting~\cite{rimoldi1996rate,grant2001rate}. However, in the presence of a privacy constraint these techniques are challenging to apply. 
In a successive decoding approach, the transmitters' messages are decoded one after another at the receiver. This approach works well in the absence of privacy constraints~\cite{winter2001capacity} because the capacity region is a polymatroid. Unfortunately, in the presence of privacy constraints, this task is challenging, even in the classical case and for only two transmitters \cite{chou2017gaussian}, because the capacity region is not known to be a polymatroid in general. With a rate-splitting approach, again, the presence of privacy constraints renders the technique challenging to apply, even in the classical case and for only two transmitters, because the rate-splitting procedure may result in negative ``rates" for some virtual users \cite{chou2018polar}.

Instead of relying on successive decoding or rate-splitting, we investigate another method (because of the challenges described above) but will still only rely on point-to-point coding techniques. Specifically, our approach in this paper relies on ideas from random binning techniques, first developed in~\cite{Csiszar1996}, which have demonstrated that three primitives are sufficient to build good codes for classical point-to-point wiretap channels. Namely, source coding with side information at the decoder \cite{slepian1973noiseless}, privacy amplification \cite{Bennett95} (which may or may not be implemented with universal hashing), and distribution approximation, i.e., the problem of creating from a random variable that is uniformly distributed, another random variable whose distribution is close (for instance with respect to relative entropy or variational distance) to a fixed target distribution, e.g.~\cite{HanBook}.
Random binning ideas has been  successfully applied to construct optimal coding schemes for point-to-point private classical communication over quantum channels~\cite{renes2011noisy} from universal hash functions (used to implement privacy amplification and distribution approximation) and schemes for source  coding with quantum side information~\cite{devetak2003classical,tomamichel2013hierarchy}. Random binning ideas have also been put forward in~\cite{yassaee2014achievability} as a means to prove the existence of good codes for classical wiretap channels, and have been applied in the context of polar coding to provide efficient  and optimal codes for several classical point-to-point wiretap channel models~\cite{chou2020explicit,chou2016polar,chou2018b}. Note that a capacity-achieving approach that separately handles the reliability constraint and the privacy constraint in the classical point-to-point wiretap channel and the classical-quantum wiretap channel has also been developped in \cite{hayashi2011exponential} and \cite{hayashi2015quantum}, respectively. \cite{hayashi2011exponential} and \cite{hayashi2015quantum} handle  the reliability constraint via  channel coding and the privacy constraint via universal hashing. We remark that the approaches in  \cite{hayashi2011exponential} and~\cite{hayashi2015quantum} differ from a random binning approach in that  \cite{hayashi2011exponential} and~\cite{hayashi2015quantum} rely on channel coding to handle the reliability constraint, whereas the random binning approach relies on source coding. Despite this difference, we believe that both approaches are interesting: The approach based on channel coding seems more natural as the wiretap channel model is a generalization of a channel coding problem, whereas the approach based on source coding uses a simpler building block, since source coding with quantum side information can be used to obtain classical-quantum   channel coding, e.g.,~\cite{renes2011noisy}.

In this paper, following random binning ideas, we establish the sufficiency of the three same primitives (source coding with quantum side information, privacy amplification, and distribution approximation) to achieve the capacity region of private classical communication over quantum multiple-access channels. Additionally, universal hashing will be sufficient to handle privacy amplification and distribution approximation. More specifically, in our coding scheme, the reliability and privacy constraints are decoupled and handled via  source coding with quantum side information at the receiver, and two-universal hash functions \cite{Carter79}, respectively. The challenge for the transmitters is to encode their private messages without the knowledge of the other users messages, and still guarantee privacy for all the messages jointly. We establish a distributed version of the leftover hash lemma against quantum side information as a tool for this task. While simultaneously smoothing the min-entropies that appears in the distributed leftover hash lemma is challenging \cite{drescher2013simultaneous}, we are still able to approximate these min-entropies by Von Neumann entropies in the case of product states. Next, to ensure reliability of the messages at the receivers we design and appropriately combine with universal hashing a multiple-access channel code designed from distributed source coding with quantum side information at the decoder. The crux of our analysis is to precisely control the joint state of the encoders output by ensuring a close trace distance between this joint state and a fixed target state in the different steps of the coding scheme, as it not only affects the rates at which the users can transmit but also the privacy guarantees. Finally, a non-trivial Fourier-Motzkin elimination that leverages submodularity properties associated with our achievable rates is performed to obtain the final expression of our achievability region.

We summarize our main contributions as follows. (i) We first derive  a regularized expression for the private classical capacity region of quantum multiple-access  channels \emph{for an arbitrary number of transmitters}. (ii) Then, we derive  a single-letter expression of the best achievable sum-rate for degradable channels by  leveraging   properties of the polymatroidal structure of the regularized capacity region.  (iii) We establish that the latter quantity is also equal to the best achievable sum-rate for quantum communication over degradable quantum multiple-access  channels. (iv) As a byproduct of independent interest, we derive a distributed version of the leftover hash lemma against quantum side information, that is used in our analysis of distributed~hashing to ensure privacy. (v)~Finally, our achievability scheme, which decouples reliability and privacy via distributed source coding and distributed hashing, establishes that the multi-user coding problem under consideration can be handled solely via point-to-point coding techniques. Namely, source coding with quantum side information between two parties and universal hashing. Even in the classical case, i.e., the classical multiple-access wiretap channel, the reduction of this multi-user coding problem to point-to-point coding techniques was only established for two transmitters but not an arbitrary number of transmitters.  %

Finally, we refer to the recent work~\cite{chakraborty2021one} for the study of a one-shot achievability scheme for the problem considered in this paper in the case of two~transmitters.

The remainder of the paper is organized as follows. We formally define the problem in Section~\ref{secps} and present our main results in Section \ref{secMR}. Before we prove our inner bound for the capacity region in Section~\ref{secachiev}, we present in Section~\ref{secprelim} preliminary results that will be used in our achievability scheme. 
Specifically, in Section \ref{secprelim}, we discuss $(i)$ distributed universal hashing against quantum side information, $(ii)$ distributed source coding with quantum side information, and $(iii)$ classical data transmission over classical-quantum multiple-access channels from distributed source coding. We prove an outer bound for the capacity region in Section \ref{secconv}. We prove our results regarding the best achievable sum-rate in Section \ref{secsumpro}. Finally, we provide concluding remarks in Section \ref{secconcl}.

\section{Notation}
For $x \in \mathbb{R}$, define $[x] \triangleq [1,\lceil x \rceil] \cap \mathbb{N}$ and $[x]^+ \triangleq \max(0,x)$. For $\mathcal{H}$, a finite-dimensional Hilbert space, let $\mathcal{P}(\mathcal{H})$ be the set of positive semi-definite operators on $\mathcal{H}$. Then, let $\mathcal{S}_{=}(\mathcal{H}) \triangleq \{ \rho \in  \mathcal{P}(\mathcal{H}) : \Tr \rho =1\}$ and  $\mathcal{S}_{\leq}(\mathcal{H}) \triangleq \{ \rho \in  \mathcal{P}(\mathcal{H}) : 0<\Tr \rho \leq 1\}$ be the set of normalized and subnormalized, respectively, quantum states. Let also $\mathcal{B}(\mathcal{H})$ denote the space of bounded linear operators on~$\mathcal{H}$. For any $\rho_{XE} \in \mathcal{S}_{\leq}( \mathcal{H}_X \otimes \mathcal{H}_E)$ and $\sigma_E \in  \mathcal{S}_{=}( \mathcal{H}_E)$, the min-entropy of $\rho_{XE}$ relative to $\sigma_E$ \cite{renner2008security}  is defined as  $H_{\min}(\rho_{XE}|\sigma_E) \triangleq \sup \left\{\lambda\in \mathbb{R}: \rho_{XE} \leq 2^{-\lambda} I_{X} \otimes \sigma_E \right\},$ where $I_{X}$ denotes the identity operator on  $\mathcal{H}_X$, 
and the max-entropy of $\rho_{E}$ \cite{renner2008security} is defined as 
$
H_{\max}(\rho_{E}) \triangleq \log \rank (\rho_E).
$ For any $\rho_{ABC} \in  \mathcal{S}_{=}( \mathcal{H}_A \otimes  \mathcal{H}_B \otimes  \mathcal{H}_C) $, define the quantum entropy $H(A)_{\rho} \triangleq - \Tr [ \rho_A \log_2 \rho_A]$, the conditional quantum entropy $H(A|B)_{\rho} \triangleq  H(AB)_{\rho} - H(B)_{\rho}$, the quantum mutual information $I(A;B)_{\rho} \triangleq H(A)_{\rho} + H(B)_{\rho} - H(AB)_{\rho}$, the quantum conditional mutual information $I(A;B|C)_{\rho} \triangleq H(A|C)_{\rho} + H(B|C)_{\rho} - H(AB|C)_{\rho}$, and the coherent information $I(A \rangle B)_{\rho} \triangleq H(B)_{\rho} - H(AB)_{\rho}$. For two probability distributions $p$ and $q$ defined over the same finite alphabet $\mathcal{X}$, define the variational distance between $p$ and $q$ as $\mathbb{V}(p,q) \triangleq \sum_{x \in \mathcal{X}} |p(x)-q(x)|$. Finally, the power set of a set $\mathcal{S}$ is denoted by $2^{\mathcal{S}}$.

\section{Problem Statement} \label{secps}
Let $L \in \mathbb{N}^*$ and define $\mathcal{L} \triangleq [ L ]$. Consider a quantum multiple-access channel $\mathcal{N}_{A'_{\mathcal{L}}\to B} : \bigotimes_{l \in \mathcal{L}} \mathcal{B}(\mathcal{H}_{A'_l}) \to \mathcal{B} (\mathcal{H}_B)  $ with $L$ transmitters, where $A'_{\mathcal{L}} \triangleq (A'_l)_{l \in \mathcal{L}}$. Let $U_{A'_{\mathcal{L}} \to BE}^{\mathcal{N}}$ be an isometric extension of the channel $\mathcal{N}_{A'_{\mathcal{L}}\to B}$ such that the complementary channel to the environment $\mathcal{N}^c_{A'_{\mathcal{L}}\to E}$ satisfies $\mathcal{N}_{A'_{\mathcal{L}}\to E}^c (\rho)   = \Tr_B [ \mathcal{U}_{A'_{\mathcal{L}} \to BE}^{\mathcal{N}} (\rho) ] $ for $\rho  \in \bigotimes_{l \in \mathcal{L}} \mathcal{B}(\mathcal{H}_{A'_l})$.

\begin{defn}\label{refdef1}
An $(n,(2^{nR_l})_{l\in \mathcal{L}})$ private classical multiple-access code for the channel $\mathcal{N}_{A'_{\mathcal{L}}\to B}$ consists of 
\begin{itemize}
\item $L$ message sets $\mathcal{M}_l \triangleq [2^{nR_l}]$, $l \in \mathcal{L}$;
\item $L$ encoding maps $\phi_l : \mathcal{M}_l \to \mathcal{B}(\mathcal{H}_{A^{\prime n}_l})$, $l \in \mathcal{L}$;
\item A decoding positive operator-valued measure (POVM) $(\Lambda_{m_{\mathcal{L}}})_{m_{\mathcal{L}} \in \mathcal{M}_{\mathcal{L}}}$, where $\mathcal{M}_{\mathcal{L}} \triangleq \bigtimes_{l\in \mathcal{L}} \mathcal{M}_l$;
\end{itemize}
and operates as follows: Transmitter $l \in \mathcal{L}$ selects a message $m_l \in \mathcal{M}_l$ and prepares the state $\rho^{m_l}_{A_l^{\prime n}} \triangleq \phi_l(m_l)$, which is sent over $\mathcal{N}_{A^{\prime n}_{\mathcal{L}}\to B^n} \triangleq (\mathcal{N}_{A'_{\mathcal{L}}\to B})^{\otimes n}$. The channel output is $\omega_{B^n}^{m_{\mathcal{L}}} \triangleq \mathcal{N}_{A^{\prime n}_{\mathcal{L}}\to B^n} (\rho^{m_{\mathcal{L}}}_{A_{\mathcal{L}}^{\prime n}})$ where $\rho_{A_{\mathcal{L}}^{\prime n}}^{m_{\mathcal{L}}} \triangleq \bigotimes_{l\in \mathcal{L}} \rho^{m_l}_{A_l^{\prime n}}$ and $m_{\mathcal{L}} \triangleq (m_l)_{l \in \mathcal{L}}$. The decoding POVM $(\Lambda_{m_{\mathcal{L}}})_{m_{\mathcal{L}} \in \mathcal{M}_{\mathcal{L}}}$ is then used at the receiver to detect the messages sent. The complementary channel output is denoted by $\omega_{E^n}^{m_{\mathcal{L}}} \triangleq \mathcal{N}^c_{A^{\prime n}_{\mathcal{L}}\to E^n} (\rho^{m_{\mathcal{L}}}_{A_{\mathcal{L}}^{\prime n}})$. 
\end{defn}

\begin{defn} \label{refdef2}
A rate-tuple $(R_l)_{l\in \mathcal{L}}$ is achievable if there exists a sequence of $(n,(2^{nR_l})_{l\in \mathcal{L}})$ private classical multiple-access codes such that for some sequence of constant states $(\sigma_{E^n})$, we have
\begin{align}
\lim_{n \to \infty} \max_{m_{\mathcal{L}} \in \mathcal{M}_{\mathcal{L}}} \Tr [(I- \Lambda_{m_{\mathcal{L}}})\omega_{B^n}^{m_{\mathcal{L}}}]& = 0, \text{ (Reliability)}\\
\lim_{n \to \infty} \max_{m_{\mathcal{L}} \in \mathcal{M}_{\mathcal{L}}}   \lVert \omega^{m_{\mathcal{L}}}_{E^n}-\sigma_{E^n} \rVert_1& = 0. \text{ (Indistinguishability)} \label{eqpriv}
\end{align}
The private classical capacity region $C_{\textup{P-MAC}}$ of a quantum multiple-access  channel $\mathcal{N}_{A'_{\mathcal{L}} \to B}$ is defined as the closure of the set of achievable rate-tuples $(R_l)_{l\in \mathcal{L}}$.
\end{defn}

\section{Main results} \label{secMR}
We first propose a regularized expression for the private classical capacity region.
\begin{thm} \label{th1}
The private classical capacity region $C_{\textup{P-MAC}}$ of a quantum multiple-access  channel $\mathcal{N}_{A'_{\mathcal{L}} \to B}$ is 
\begin{align*}
C_{\textup{P-MAC}} (\mathcal{N}) = \textup{cl} \left(  \bigcup_{n=1}^{\infty} \frac{1}{n} \mathcal{P} (\mathcal{N}^{\otimes n}) \right),
\end{align*}
where $\textup{cl}$ denotes the closure operator and $\mathcal{P} (\mathcal{N})$ is the set of rate-tuples  $(R_l)_{l\in \mathcal{L}}$ that satisfy
$$ R_{\mathcal{S}} \triangleq \sum_{l\in \mathcal{S} }R_{l} \leq [I(X_{\mathcal{S}};B |X_{\mathcal{S}^c})_{\rho} - I(X_{\mathcal{S}};E)_{\rho}]^+, \forall \mathcal{S} \subseteq \mathcal{L},$$
for some classical-quantum state $\rho_{X_{\mathcal{L}} A'_{\mathcal{L}}}$ of the form
$$
\rho_{X_{\mathcal{L}} A'_{\mathcal{L}}} \triangleq \bigotimes_{l \in \mathcal{L}} \left( \sum_{x_{l}} p_{X_l} (x_l) \ket{x_l} \bra{x_l}_{X_l}  \otimes \rho_{A'_l}^{x_l}\right),
$$
and $\rho_{X_{\mathcal{L}}BE} \triangleq  \mathcal{U}^{\mathcal{N}}_{A'_{\mathcal{L}} \to BE} (\rho_{X_{\mathcal{L}} A'_{\mathcal{L}}} )$ with ${U}^{\mathcal{N}}_{A'_{\mathcal{L}} \to BE}$ an isometric extension of $\mathcal{N}_{A'_{\mathcal{L}} \to B}$, and the notation $X_{\mathcal{S}} \triangleq (X_l)_{l\in\mathcal{S}}$ for any $\mathcal{S} \subseteq \mathcal{L}$.
\end{thm}
\begin{proof}
The achievability and converse are proved in Sections~\ref{secachiev} and \ref{secconv}, respectively.
\end{proof}

In the next result, for the case of degradable channels, we propose a single-letter expression for the best achievable sum-rate in the private classical capacity region.

\begin{thm} \label{th2}
Consider a degradable quantum multiple-access  channel $\mathcal{N}_{A'_{\mathcal{L}} \to B}$, i.e., there exists a channel $\mathcal{D}_{B \to E}$ such that $\mathcal{D}_{B \to E} \circ  \mathcal{N}_{A'_{\mathcal{L}} \to B}=\mathcal{N}^c_{A'_{\mathcal{L}} \to E} $. Define $C_{\textup{P-MAC}}^{\textup{sum}}$ as the supremum of all achievable sum-rates in $C_{\textup{P-MAC}} (\mathcal{N})$. Then, we have
$$C_{\textup{P-MAC}}^{\textup{sum}} (\mathcal{N}) = P_{\textup{MAC}}^{\textup{sum}} (\mathcal{N}),$$ with
\begin{align}
 P_{\textup{MAC}}^{\textup{sum}} (\mathcal{N}) \triangleq \max_{\rho}  [I(X_{\mathcal{L}};B )_{\rho} - I(X_{\mathcal{L}};E)_{\rho}]^+, \label{eqPmacsum}
\end{align}
 where the maximization is over classical-quantum states that have the same form as in Theorem~\ref{th1}.
\end{thm}
\begin{proof}
See Section \ref{secsumpro}.
\end{proof}

We now propose another single-letter characterization of $C_{\textup{P-MAC}}^{\textup{sum}}$ for degradable channels. We first define the quantity $Q^{\textup{sum}}_{\textup{MAC}}$.  
\begin{defn} \label{defQsum}
Consider a quantum multiple-access  channel $\mathcal{N}_{A_{\mathcal{L}}' \to B}$.	Define 
\begin{align}
Q^{\textup{sum}}_{\textup{MAC}} (\mathcal{N})\triangleq \max_{\phi_{A_{\mathcal{L}} A'_{\mathcal{L}}}} I(A_{\mathcal{L}} \rangle B)_{\rho}, \label{eqqsum}
\end{align}
where the maximization is over states of the form $ \phi_{A_{\mathcal{L}} A'_{\mathcal{L}}} \triangleq \bigotimes_{l \in \mathcal{L}} \phi_{A_lA'_l}$ with  $\phi_{A_lA'_l}$, $l\in \mathcal{L}$, a pure state,  and 
$
\rho_{A_{\mathcal{L}} B } \triangleq \mathcal{N}_{A'_{\mathcal{L}} \to B} (\phi_{A_{\mathcal{L}} A'_{\mathcal{L}}}).
$
\end{defn}
Note that by \cite{yard2008capacity}, $ \lim_{n\to \infty} \frac{1}{n} Q^{\textup{sum}}_{\textup{MAC}}  (\mathcal{N}^{\otimes n})$ is a regularized expression for the largest achievable sum-rate for quantum communication over quantum multiple-access channels.

\begin{thm} \label{th3}
Consider a degradable quantum multiple-access  channel $\mathcal{N}_{A'_{\mathcal{L}} \to B}$. Then, we have
$$C_{\textup{P-MAC}}^{\textup{sum}} (\mathcal{N})= Q_{\textup{MAC}}^{\textup{sum}} (\mathcal{N}).$$ 
\end{thm}
\begin{proof}
See Section \ref{secsumpro}.
\end{proof}
Note that in the case of point-to-point channels Theorem \ref{th3} recovers the result in \cite[Th. 2]{smith2008private}.

\section{Preliminary results} \label{secprelim}
We establish in this section preliminary results that we will use to show in Section \ref{secachiev} the achievability part of Theorem~\ref{th1}.
\subsection{Distributed leftover hash lemma against quantum side information} \label{sec:lohlq}

Define $\mathcal{L} \triangleq [L]$. Consider the random variables $X_{\mathcal{L}} \triangleq (X_l)_{l\in\mathcal{L}}$, defined over the Cartesian product $\mathcal{X}_{\mathcal{L}} \triangleq \bigtimes_{l \in \mathcal{L}} \mathcal{X}_l$ with probability distribution $p_{X_{\mathcal{L}}}$, and a quantum system $E$ whose state depends on $X_{\mathcal{L}}$, described by the following classical-quantum state:
\begin{align} \label{rhoxe}
\rho_{X_{\mathcal{L}}E} 
& \triangleq \sum_{x_{\mathcal{L}} \in \mathcal{X}_{\mathcal{L}}}  \ket{x_{\mathcal{L}}}\!\bra{x_{\mathcal{L}}} \otimes \rho_E^{x_{\mathcal{L}}},
\end{align} 
where $\ket{x_{\mathcal{L}}}\!\bra{x_{\mathcal{L}}} \triangleq \bigotimes_{l \in \mathcal{L}} \ket{x_{l}}\!\bra{x_{l}}$ and $\rho_E^{x_{\mathcal{L}}} \triangleq p_{X_{\mathcal{L}}}(x_{\mathcal{L}}) \bar{\rho}_E^{x_{\mathcal{L}}}$ with $\bar{\rho}_E^{x_{\mathcal{L}}}$ the state of the system $E$ conditioned on the realization $x_{\mathcal{L}}$. 
Next, consider $F_l:\mathcal{X}_l \to \{0,1\}^{r_l}$ a hash function chosen uniformly at random in a family $\mathcal{F}_l$, $l \in \mathcal{L}$, of two-universal hash functions \cite{Bennett95}, i.e., $$\forall x_l,x_l' \in \mathcal{X}_l, x_l\neq x_l' \implies \mathbb{P}[F_l(x_l) = F_l(x_l')] \leq 2^{-r_l}.$$ For any $\mathcal{S} \subseteq \mathcal{L}$, define $\mathcal{X}_{\mathcal{S}} \triangleq \bigtimes_{l \in \mathcal{S}} \mathcal{X}_l$, $F_{\mathcal{S}} \triangleq (F_{l})_{l\in \mathcal{S}}$, $\mathcal{F}_{\mathcal{S}} \triangleq \bigtimes_{l\in \mathcal{S}} \mathcal{F}_{l}$,  $\mathcal{A}_{\mathcal{S}} \triangleq \bigtimes_{l\in \mathcal{S}} \{0,1\}^{r_l}$, and for $a_{\mathcal{S}} \in \mathcal{A}_{\mathcal{S}}$, $f_{\mathcal{S}} \in \mathcal{F}_{\mathcal{S}}$, $f^{-1}_{\mathcal{S}}(a_{\mathcal{S}}) \triangleq \{ x_{\mathcal{S}} \in \mathcal{X}_{\mathcal{S}} : f_l(x_l) = a_l, \forall l \in \mathcal{S}\}$. The hash functions outputs $f_{\mathcal{L}}(x_{\mathcal{L}})\triangleq (f_l(x_l))_{l\in\mathcal{L}}$, the state of the quantum system, and the choice of the functions $f_{\mathcal{L}}$ are
 described by the following operator
 \begin{align}
 &\rho_{F_{\mathcal{L}}(X_{\mathcal{L}})E F_{\mathcal{L}}} \nonumber\\
 &\phantom{-}\triangleq \frac{1}{|\mathcal{F}_{\mathcal{L}}|} \sum_{f_{\mathcal{L}} \in \mathcal{F}_{\mathcal{L}}}   \sum_{a_{\mathcal{L}} \in \mathcal{A}_{\mathcal{L}}} \ket{a_{\mathcal{L}}}\!\bra{a_{\mathcal{L}}} \otimes {\rho}_E^{f_{\mathcal{L}},a_{\mathcal{L}}}  \otimes  \ket{f_{\mathcal{L}}}\!\bra{f_{\mathcal{L}}}, \label{rhoxelhl}
 \end{align}
where ${\rho}_E^{f_{\mathcal{L}},a_{\mathcal{L}}} \triangleq \sum_{x_{\mathcal{L}} \in f_{\mathcal{L}}^{-1}(a_{\mathcal{L} })} {\rho}_E^{x_{\mathcal{L}}} $, $\ket{a_{\mathcal{L}}}\!\bra{a_{\mathcal{L}}} \triangleq \bigotimes_{l \in \mathcal{L}} \ket{a_{l}}\!\bra{a_{l}}$, and $\ket{f_{\mathcal{L}}}\!\bra{f_{\mathcal{L}}} \triangleq \bigotimes_{l \in \mathcal{L}} \ket{f_{l}}\!\bra{f_{l}}$.

\begin{lem}[Distributed leftover hash lemma] \label{lemlhl}
Let $\rho_U$ be the fully mixed state on $\mathcal{H}_{F_{\mathcal{L}}(X_{\mathcal{L}})}$. Define for any $\mathcal{S} \subseteq \mathcal{L}$,  $r_{\mathcal{S}} \triangleq \sum_{s\in\mathcal{S}} r_s$. For any $\sigma_E \in \mathcal{S}_{=}(\mathcal{H}_E)$, we have
\begin{align*}
\lVert \rho_{F_{\mathcal{L}}(X_{\mathcal{L}})E F_{\mathcal{L}}}  -  \rho_U \otimes \rho_{EF_{\mathcal{L}}} \rVert_1 \leq \sqrt{  \sum_{\substack{\mathcal{S} \subseteq \mathcal{L} \\ \mathcal{S} \neq \emptyset}}    2^{r_{\mathcal{S}}- H_{\min}(\rho_{X_{\mathcal{S}}E}|\sigma_E)}}.
\end{align*}
\end{lem}
\begin{proof}
See Appendix \ref{app_1}.
\end{proof}
Note that a similar lemma was known in the classical case, e.g., \cite{wullschleger2007oblivious}, and had found applications to oblivious transfer~\cite{nascimento2008oblivious,wullschleger2007oblivious,Chou21}, secret generation \cite{chou2019secret,chou2019biometric,chou2018secret}, and multiple-access channel resolvability \cite{sultana2020explicit}. We are now interested in deriving a distributed leftover hash lemma for product states. We will use the following result on product probability distributions, which is a kind of asymptotic equipartition property (AEP) that holds simultaneously for a set of min-entropies.

\begin{lem} \label{lemprod} 
Consider the random variables $X^n_{\mathcal{L}} \triangleq (X_l)_{l\in \mathcal{L}}$, $Y^n$ defined over $\mathcal{X}^n_{\mathcal{L}}\times \mathcal{Y}^n$ with probability distribution $p_{X^n_{\mathcal{L}}Y^n} \triangleq \prod_{i=1}^n p_{X_{\mathcal{L}}Y}$. In this lemma, let  $H(\cdot)$ denote the Shannon entropy for random variables following $p_{X_{\mathcal{L}}Y}$ or its marginals. For any $\epsilon>0$, there exists a subnormalized non-negative function $q_{X^n_{\mathcal{L}}Y^n}$ defined over $\mathcal{X}^n_{\mathcal{L}} \times \mathcal{Y}^n $ such that $\mathbb{V}(p_{X^n_{\mathcal{L}}Y^n},q_{X^n_{\mathcal{L}}Y^n})   \leq \epsilon$ and
\begin{align*}
\forall \mathcal{S} \subseteq \mathcal{L}, H_{\min}(q_{X^n_{\mathcal{S}}Y^n}) &\geq n H(X_{\mathcal{S}}Y) - n \delta_{\mathcal{S}}(n),\\
H_{\max}(q_{Y^n})& \leq n H(Y) + n \delta(n),
\end{align*}
where $\delta_{\mathcal{S}}(n) \triangleq \log(|\mathcal{X}_{\mathcal{S}}||\mathcal{Y}| +3) \sqrt{\frac{2}{n} (L+1+ \log (\frac{1}{\epsilon}))}$, $\forall \mathcal{S} \subseteq \mathcal{L}$, $\delta(n) \triangleq \log(|\mathcal{Y}| +3) \sqrt{\frac{2}{n}(1+  \log (\frac{1}{\epsilon}))}$.
\end{lem}

\begin{proof}
See Appendix \ref{app_2}.
\end{proof}
From Lemmas \ref{lemlhl} and \ref{lemprod}, we then obtain the following result. 
\begin{lem}[Distributed leftover hash lemma for product states] \label{lemlhlprod}
Consider the product state $\rho_{X^n_{\mathcal{L}}E^n} \triangleq \rho_{X_{\mathcal{L}}E}^{\otimes n} $, where $\rho_{X_{\mathcal{L}}E}$ is defined in  \eqref{rhoxe}. With the same notation as in Lemma \ref{lemlhl}, we have
\begin{align*}
& \lVert \rho_{F_{\mathcal{L}}(X^n_{\mathcal{L}})E^n F_{\mathcal{L}}}  -  \rho_U \otimes \rho_{E^nF_{\mathcal{L}}} \rVert_1 \\
& \phantom{-}\leq 2 \epsilon  + \sqrt{  \sum_{ \substack{\mathcal{S} \subseteq \mathcal{L}\\ \mathcal{S} \neq \emptyset} }   2^{r_{\mathcal{S}} - n H({X_{\mathcal{S}}|E})_{\rho} +n( \delta_{\mathcal{S}}(n) + \delta(n))} },
\end{align*}
where $\delta_{\mathcal{S}}(n) \triangleq \log(|\mathcal{X}_{\mathcal{S}}|d_E +3) \sqrt{\frac{2}{n} (L+ 1+ \log (\frac{1}{\epsilon}))}$,  $\delta(n) \triangleq \log(d_E +3) \sqrt{\frac{2}{n}(1+ \log (\frac{1}{\epsilon}))}$, with $d_E \triangleq \dim \mathcal{H}_E$.
\end{lem}

\begin{proof}
See Appendix \ref{app_3}.
\end{proof}

\subsection{Distributed classical source coding with quantum side information}

Consider $X_{\mathcal{L}} \triangleq (X_l)_{l\in\mathcal{L}}$, defined over $\mathcal{X}_{\mathcal{L}} \triangleq \bigtimes_{l \in \mathcal{L}} \mathcal{X}_l$ with probability distribution $p_{X_{\mathcal{L}}}$, and a quantum system $B$ whose state depends on the random variable $X_{\mathcal{L}}$, described by the following classical-quantum state
\begin{align*} 
\rho_{X_{\mathcal{L}}B} 
& \triangleq \sum_{x_{\mathcal{L}} \in \mathcal{X}_{\mathcal{L}}}  \ket{x_{\mathcal{L}}}\!\bra{x_{\mathcal{L}}} \otimes \rho_B^{x_{\mathcal{L}}},
\end{align*} 
where $\rho_B^{x_{\mathcal{L}}} \triangleq p_{X_{\mathcal{L}}}(x_{\mathcal{L}}) \bar{\rho}_B^{x_{\mathcal{L}}}$ with $\bar{\rho}_B^{x_{\mathcal{L}}}$ the state of the system $B$ conditioned on the realization $x_{\mathcal{L}}$, and we have used the same notation as in Section \ref{sec:lohlq}. 
\begin{defn} \label{defsw}
A $(2^{nR_l})_{l\in\mathcal{L}}$ distributed source code for a  classical-quantum product state $\rho_{X_{\mathcal{L}}B}^{\otimes n}$ consists of
\begin{itemize}
\item $L$ sets $\mathcal{C}_l \triangleq [2^{nR_l}]$, $l\in\mathcal{L}$;
\item $L$ encoders $g_l:\mathcal{X}^n_l \to \mathcal{C}_l$, $l\in\mathcal{L}$;
\item One decoder $h: \mathcal{S}_=(\mathcal{H}_{B^n}) \times \mathcal{C}_{\mathcal{L}}\to \mathcal{X}^n_{\mathcal{L}}$, where $\mathcal{C}_{\mathcal{L}} \triangleq \bigtimes_{l \in \mathcal{L}} \mathcal{C}_l$. %
\end{itemize} 
A rate-tuple $(R_l)_{l\in\mathcal{L}}$ is said to be achievable  when the average error probability $P_e(n) \triangleq \sum_{x_{\mathcal{L}}^n \in \mathcal{X}_{\mathcal{L}}^n} p_{X^n_{\mathcal{L}}}(x^n_{\mathcal{L}}) \mathbb{P} \left[h ( \bar{\rho}_{B^n}^{x^n_{\mathcal{L}}} ,g_{\mathcal{L}}(x^n_{\mathcal{L}}) )\neq x^n_{\mathcal{L}} \right]$ satisfies $\lim_{n\to \infty}P_e(n) = 0$, 
where for all $x_{\mathcal{L}}^n \in \mathcal{X}_{\mathcal{L}}^n$, $g_{\mathcal{L}}(x^n_{\mathcal{L}}) \triangleq (g_l(x^n_l))_{l\in\mathcal{L}}$. Let $\mathcal{C} (\rho_{X_{\mathcal{L}}B} )$ be the set of all achievable rate-tuples.
\end{defn}

\begin{lem}[\cite{winter1999coding}] \label{lemSW}
We have 
$$
\mathcal{C} (\rho_{X_{\mathcal{L}}B} )  = \{ (R_l)_{l\in \mathcal{L}} : R_{\mathcal{S}} \geq H(X_{\mathcal{S}}|X_{\mathcal{S}^c}B)_{\rho} , \forall \mathcal{S} \subseteq \mathcal{L}\}.
$$
\end{lem}
Note that the set $\{ (R_l)_{l\in\mathcal{L}} : R_{\mathcal{S}} \geq H(X_{\mathcal{S}}|X_{\mathcal{S}^c}B)_{\rho} , \forall \mathcal{S} \subseteq \mathcal{L}\}$  associated with the set function $\mathcal{S} \mapsto H(X_{\mathcal{S}}|X_{\mathcal{S}^c}B)_{\rho} $ defines a contrapolymatroid. Using the fact that its dominant face, i.e., $\{ (R_l)_{l\in\mathcal{L}} \in \mathcal{C} (\rho_{X_{\mathcal{L}}B} ):  R_{\mathcal{L}} = H(X_{\mathcal{L}}|B)_{\rho} \}$ is  the convex hull of its extreme points \cite{edmonds2003submodular}, one can easily verify that the region $\mathcal{C} (\rho_{X_{\mathcal{L}}B} ) $ is achievable using source coding with quantum side information for two parties \cite{devetak2003classical} and time-sharing.  This is exactly the coding technique employed in \cite{winter1999coding} to prove Lemma \ref{lemSW}.

\subsection{Multiple-access channel coding from distributed source coding} \label{secMAC}
Consider $L$ finite sets  $\mathcal{U}_l$, $l \in \mathcal{L}$, such that  $|\mathcal{U}_l| = 2^{R^{\textup{U}}_l}$ for some $R^{\textup{U}}_l \in \mathbb{R}_+$ and define $\mathcal{U}_{\mathcal{L}} \triangleq \bigtimes_{l\in\mathcal{L}} \mathcal{U}_l$. Consider a classical-quantum multiple-access channel, i.e., a map $W: \mathcal{U}_{\mathcal{L}} \to \mathcal{S}_{=}(\mathcal{H}_B)$, which maps $u_{\mathcal{L}} \in \mathcal{U}_{\mathcal{L}}$ to the state $\bar{\rho}_B^{u_{\mathcal{L}}} \in \mathcal{S}_{=}(\mathcal{H}_B) $.
Let $\rho_{U_{\mathcal{L}}B} \triangleq \frac{1}{|\mathcal{U}_{\mathcal{L}}|} \sum_{u_{\mathcal{L}} \in \mathcal{U}_{\mathcal{L}}} \ket{u_{\mathcal{L}}}\bra{u_{\mathcal{L}}} \otimes \bar{\rho}_B^{u_{\mathcal{L}}}$ describe the input and output of $W$ when the input $U_{\mathcal{L}}$ is uniformly distributed over $\mathcal{U}_{\mathcal{L}}$, and where we have used the notation $\ket{u_{\mathcal{L}}}\bra{u_{\mathcal{L}}} \triangleq \bigotimes_{l \in \mathcal{L}} \ket{u_{l}}\bra{u_{l}}$.

\begin{lem}[Multiple-access channel coding from distributed source coding] \label{lemMac}
Consider $L$ uniformly distributed messages $(M_l)_{l \in \mathcal{L}} \in \mathcal{M}_{\mathcal{L}} \triangleq \bigtimes_{l \in \mathcal{L}} \mathcal{M}_l $, where $\mathcal{M}_l \triangleq  [2^{nR_l}]$ for some $R_l \in \mathbb{R}_+$, $l \in \mathcal{L}$.  If there exists  a $(2^{nR^{\textup{DC}}_l})_{l\in\mathcal{L}}$ distributed source code (as defined in Definition \ref{defsw}) for the classical-quantum product state $\rho_{U_{\mathcal{L}}B}^{\otimes n}$, then there exist $L$ encoders $e_l: \mathcal{M}_l\to \mathcal{U}_l^n$, ${l\in\mathcal{L}}$, and one decoder $d:\mathcal{S}_{=}(\mathcal{H}_{B^n}) \to \mathcal{M}_{\mathcal{L}}$ such that one can choose $ R_l = R^{\textup{U}}_l - R^{\textup{DC}}_l$ as $n\to \infty$, $l\in \mathcal{L}$,  and $\lim_{n\to \infty}\mathbb{P}[d( \bar{\rho}_{B^n}^{e_{\mathcal{L}}(M_{\mathcal{L}})} ) \neq M_{\mathcal{L}}] = 0$, where $e_{\mathcal{L}}(M_{\mathcal{L}}) \triangleq (e_{l}(M_{l}))_{l\in \mathcal{L}}$.
\end{lem}
\begin{proof}
See Appendix \ref{App_5}.
\end{proof}
Note that this lemma recovers \cite[Lemma 2]{renes2011noisy}, which treats the case of point-to-point channels.

\section{Achievability of Theorem \ref{th1}} \label{secachiev}
Consider a classical-quantum multiple-access wiretap channel, i.e., a map $W: \mathcal{X}_{\mathcal{L}} \to \mathcal{S}_{=}(\mathcal{H}_B \otimes \mathcal{H}_E)$, which maps $x_{\mathcal{L}} \in \mathcal{X}_{\mathcal{L}}$ to $\bar{\rho}_{BE}^{x_{\mathcal{L}}}\in \mathcal{S}_{=}(\mathcal{H}_B \otimes \mathcal{H}_E)$. The achievability part of Theorem \ref{th1} reduces to another achievability result (with a slight adaptation of Definitions \ref{refdef1}, \ref{refdef2}) for this classical-quantum multiple-access wiretap channel. Specifically, we show in this section that, for any probability distribution $p_{X_{\mathcal{L}}} \triangleq \prod_{l \in \mathcal{L}} p_{X_l}$, the following region is achievable 
\begin{align*}
&\mathcal{R}(W,p_{X_{\mathcal{L}}})\\
&  \triangleq \! \{ (R_{l \in \mathcal{L}})\! : \! R_{\mathcal{S}}  \leq \! [I(X_{\mathcal{S}};B |X_{\mathcal{S}^c})_{\rho} \! - \! I(X_{\mathcal{S}};E)_{\rho}]^+\!, \forall \mathcal{S} \subseteq \mathcal{L}\},
\end{align*}
where $\rho_{X_{\mathcal{L}}BE} \triangleq \sum_{x_{\mathcal{L}}} p_{X_{\mathcal{L}}}(x_{\mathcal{L}}) \ket{x_{\mathcal{L}}}\bra{x_{\mathcal{L}}} \otimes \bar{\rho}_{BE}^{x_{\mathcal{L}}}$. Note that, compared to the setting of Section~\ref{secps}, the signal states sent by the transmitters are now part of the channel definition. Hence, achievability of $\mathcal{R}(W,p_{X_{\mathcal{L}}})$ and regularization lead to the achievability part of Theorem \ref{th1}.

\subsection{Coding scheme} \label{secCS}

The main idea of the coding scheme is to combine distributed source coding and distributed randomness extraction to emulate a random binning-like proof. We proceed in three~steps.

\textbf{Step 1}: We create a stochastic channel that simulates the inversion of multiple hash functions while approximating the joint distribution of the inputs and outputs of the hash functions. Approximating this joint distribution is crucial for the message indistinguishability analysis. In the special case of a single hash function, this operation is referred to as shaping in \cite{renes2011noisy} and distribution approximation in \cite{chou2016polar}.

Consider $X_{\mathcal{L}}^n$ distributed according to some arbitrary product distribution $p_{X^n_{\mathcal{L}}} \triangleq \prod_{l \in \mathcal{L}} p_{X_l^n}$, and $L$ two-universal hash functions $F_{\mathcal{L}}$ uniformly distributed over $\mathcal{F}_{\mathcal{L}}$, where we use the same notation as in Section \ref{sec:lohlq}. The output lengths  of the hash functions, denoted by $(nR^{\textup{U}}_l)_{l\in \mathcal{L}}$, will be defined later. 
Let $\widetilde{W}_{\mathcal{L}}$ be the channel  described by the conditional probability distribution
$p_{X^n_{\mathcal{L}}|F_{\mathcal{L}}(X^n_{\mathcal{L}})F_{\mathcal{L}}} \triangleq \prod_{l \in \mathcal{L}}p_{X^n_{l}|F_{l}(X^n_{l})F_{l} } $ and $\widetilde{W}_l$ be the channel  described by the conditional probability distribution
$p_{X^n_{l}|F_{l}(X^n_{l})F_{l}} $, $l \in \mathcal{L}$. For $l\in\mathcal{L}$, let $U^n_{l}$ be uniformly distributed over $\mathcal{U}^n_l \triangleq  [2^{nR^{\textup{U}}_l}]$, and define 
\begin{align} \label{eqptilde}
 \widetilde{p}_{X^n_{\mathcal{L}}U^n_{\mathcal{L}}F_{\mathcal{L}}} \triangleq p_{X^n_{\mathcal{L}}|F_{\mathcal{L}}(X^n_{\mathcal{L}})F_{\mathcal{L}}} p_{U^n_{\mathcal{L}}} p_{F_{\mathcal{L}}} ,
 \end{align}
where $p_{U^n_{\mathcal{L}}}$ is the uniform distribution over $\mathcal{U}^n_{\mathcal{L}}$ with the same notation as in Section \ref{secMAC}.
 Hence, $\widetilde{p}_{X^n_{\mathcal{L}}U^n_{\mathcal{L}}F_{\mathcal{L}}}$ denotes the joint probability distribution of the input $(U^n_{\mathcal{L}},F_{\mathcal{L}})$ and output $\widetilde{X}^n_{\mathcal{L}} \triangleq \widetilde{W}_{\mathcal{L}}(U^n_{\mathcal{L}},F_{\mathcal{L}})$ of the channel $\widetilde{W}_{\mathcal{L}}$.   To simplify notation in the following, we write  $\widetilde{W}_{\mathcal{L}}(U^n_{\mathcal{L}})$ instead of $\widetilde{W}_{\mathcal{L}}(U^n_{\mathcal{L}},F_{\mathcal{L}})$ by redefining $\widetilde{W}_{\mathcal{L}}$ and including   $F_{\mathcal{L}}$ in its definition. 

\textbf{Step 2}: Using Lemma \ref{lemMac}, we construct a multiple-access channel code for jointly uniform input distributions (in the absence of any privacy constraint) for the channel $W \circ \widetilde{W}_{\mathcal{L}}$.

Let $m \in \mathbb{N}$. By Lemma \ref{lemSW}, there exists a $(2^{mnR^{\textup{DC}}_l})_{l\in\mathcal{L}}$ distributed source code (as defined in Definition \ref{defsw}) for the classical-quantum product state $\widetilde{\rho}_{U_{\mathcal{L}}^nB^n}^{\otimes m}$, where  \begin{align} \label{eqrhotildub}\widetilde{\rho}_{U_{\mathcal{L}}^nB^n}\triangleq \frac{1}{|\mathcal{U}^n_{\mathcal{L}}|} \sum_{u^n_{\mathcal{L}} \in \mathcal{U}^n_{\mathcal{L}}} \ket{u^n_{\mathcal{L}}}\bra{u^n_{\mathcal{L}}} \otimes \bar{\rho}_{B^n}^{\widetilde{W}_{\mathcal{L}}(u^n_{\mathcal{L}})},\end{align} and where $(nR^{\textup{DC}}_l)_{l\in\mathcal{L}}$ belongs to $\mathcal{C}(\widetilde{\rho}_{U_{\mathcal{L}}^nB^n})$. Then, by Lemma~\ref{lemMac}, there exist $L$ encoders
 $e_l: \mathcal{M}^m_l\to \mathcal{U}_l^{mn}, {l\in\mathcal{L}},$ and one decoder $d:\mathcal{S}_{=}(\mathcal{H}_{B^{mn}}) \to \mathcal{M}^m_{\mathcal{L}},$ where we have defined for $l\in \mathcal{L}$, $\mathcal{M}^m_l \triangleq [2^{mnR_l}]$ such that  $ R_l = R^{\textup{U}}_l - R^{\textup{DC}}_l$ as $m\to \infty$,  and 
 \begin{align} \label{eqdecoder}
 \lim_{m\to \infty}\mathbb{P}\left[d\left( \bar{\rho}_{B^{mn}}^{\widetilde{W}_{\mathcal{L}}^{\otimes m}(e_{\mathcal{L}}(M^m_{\mathcal{L}}))} \right) \neq M^m_{\mathcal{L}}\right] = 0,
 \end{align}
  with $e_{\mathcal{L}}(M^m_{\mathcal{L}}) \triangleq (e_{l}(M^m_{l}))_{l\in \mathcal{L}}$.

\textbf{Step 3}: We combine Step 1 and Step 2 to define our encoders and decoder for the classical-quantum multiple-access wiretap channel. Specifically, the encoders are  defined~as  \begin{align} \label{eqencoderL}
\phi_{l} : M^m_{l} \mapsto \widetilde{W}_l^{\otimes m} (e_{l} (M^m_{l})), l \in \mathcal{L}, 
\end{align}
 and the decoder is defined as  \begin{align} \label{eqdcodpsi}
\psi: \bar{\rho}_{B^{mn}}^{\phi_{\mathcal{L}}(M^m_{\mathcal{L}})} \mapsto d(\bar{\rho}_{B^{mn}}^{\phi_{\mathcal{L}}(M^m_{\mathcal{L}})}), 
\end{align}
 where $\phi_{\mathcal{L}}(M^m_{\mathcal{L}}) \triangleq (\phi_l(M^m_l))_{l\in \mathcal{L}}$.
\begin{rem}
In Step 2, Lemma \ref{lemSW} cannot be directly applied to $\widetilde{\rho}_{U_{\mathcal{L}}^nB^n}$ as it is not a product state. %
\end{rem}
\subsection{Coding scheme analysis}
\subsubsection{Average reliability}
We have 
\begin{align} \label{eqproofrelia}
&\mathbb{P} \left[ \psi (\bar{\rho}_{B^{mn}}^{\phi_{\mathcal{L}}(M^m_{\mathcal{L}})} ) \neq M^m_{\mathcal{L}}\right] \nonumber \\
& = \mathbb{P} \left[ d (\bar{\rho}_{B^{mn}}^{ \widetilde{W}_{\mathcal{L}}^{\otimes m} (e_{\mathcal{L}}(M^m_{\mathcal{L}}))} ) \neq M^m_{\mathcal{L}}\right] \xrightarrow{m \to \infty}	0,
\end{align}
where the equality holds by definition of $\psi$ and $(\phi_l)_{l\in \mathcal{L}}$ in~\eqref{eqencoderL}, \eqref{eqdcodpsi}, and the limit holds by \eqref{eqdecoder}.

\subsubsection{Average message indistinguishability}
%
%
Note that 
by a random choice of the encoder in the proof of Lemma~\ref{lemMac}, $e_{\mathcal{L}}(M^m_{\mathcal{L}})$ is uniformly distributed, hence, $\widetilde{W}_{\mathcal{L}}^{\otimes m} (e_{\mathcal{L}}(M^m_{\mathcal{L}}))$ follows a product distribution and $\widetilde{\rho}_{e_{\mathcal{L}}(M^m_{\mathcal{L}})E^{mn}F^m_{\mathcal{L}}}$ is a product state, which one can write $\widetilde{\rho}_{e_{\mathcal{L}}(M^m_{\mathcal{L}})E^{mn}F^m_{\mathcal{L}}} = \widetilde{\rho}^{\otimes m}_{U^n_{\mathcal{L}}E^{n}F_{\mathcal{L}}}$, where
\begin{align} 
\widetilde{\rho}_{U^n_{\mathcal{L}} E^{n}F_{\mathcal{L}}}
& \triangleq  \smash{ \sum_{{f}_{\mathcal{L}} }   \sum_{{u}^n_{\mathcal{L}} }  \sum_{{x}^n_{\mathcal{L}} } } \widetilde{p}_{X^{n}_{\mathcal{L}}U^n_{\mathcal{L}} F_{\mathcal{L}} } ({x}^n_{\mathcal{L}},{u}^n_{\mathcal{L}},{f}_{\mathcal{L}} )  \nonumber \\
&\phantom{--------}\ket{{u}^n_{\mathcal{L}}}\!\bra{{u}^n_{\mathcal{L}}} \otimes  \bar{\rho}_{E^{n}}^{{x}^n_{\mathcal{L}}} \otimes  \ket{{f}_{\mathcal{L}}}\!\bra{{f}_{\mathcal{L}}}.\label{eqbockE}
\end{align}
Next, define the following classical-quantum state
\begin{align} 
{\rho}_{F_{\mathcal{L}}(X^n_{\mathcal{L}})E^{n}F_{\mathcal{L}}}
& \triangleq \smash{ \sum_{{f}_{\mathcal{L}} }   \sum_{{u}^n_{\mathcal{L}} }  \sum_{{x}^n_{\mathcal{L}} } {p}_{X^{n}_{\mathcal{L}}F_{\mathcal{L}}(X^n_{\mathcal{L}})F_{\mathcal{L}} } ({x}^n_{\mathcal{L}},{u}^n_{\mathcal{L}},{f}_{\mathcal{L}} ) } \nonumber \\
&\phantom{-------} \ket{{u}^n_{\mathcal{L}}}\!\bra{{u}^n_{\mathcal{L}}} \otimes  \bar{\rho}_{E^{n}}^{{x}^n_{\mathcal{L}}} \otimes  \ket{{f}_{\mathcal{L}}}\!\bra{{f}_{\mathcal{L}}}. \label{eqbockE2}
\end{align}
Then, for  $\bar{\rho}_{U}$ the fully mixed state on $\mathcal{H}_{U^n_{\mathcal{L}}}$ and ${\rho}_{U}$ the fully mixed state on $\mathcal{H}_{M_{\mathcal{L}}}$, we have 
\begin{align*}
& \lVert \widetilde{\rho}_{ M^m_{\mathcal{L}} E^{mn} F^m_{\mathcal{L}}}  -  {\rho}_{U}^{\otimes m} \otimes \widetilde{\rho}_{E^{mn}F^m_{\mathcal{L}}} \rVert_1  \displaybreak[0] \\
& \leq  \lVert \widetilde{\rho}_{e_{\mathcal{L}}(M^m_{\mathcal{L}})E^{mn} F^m_{\mathcal{L}}}  -  \bar{\rho}_{U}^{\otimes m} \otimes \widetilde{\rho}_{E^{mn}F^m_{\mathcal{L}}} \rVert_1  \displaybreak[0] \\
& = \lVert \widetilde{\rho}^{\otimes m}_{U^n_{\mathcal{L}}E^n F_{\mathcal{L}}}  -  \bar{\rho}_{U}^{\otimes m} \otimes \widetilde{\rho}^{\otimes m}_{E^nF_{\mathcal{L}}} \rVert_1  \\
&\stackrel{(a)}  \leq m \lVert \widetilde{\rho}_{U^n_{\mathcal{L}}E^n F_{\mathcal{L}}}  -  \bar{\rho}_{U} \otimes \widetilde{\rho}_{E^nF_{\mathcal{L}}} \rVert_1  \\
& \stackrel{(b)} \leq m( \lVert \widetilde{\rho}_{U^n_{\mathcal{L}}E^n F_{\mathcal{L}}}  -  {\rho}_{F_{\mathcal{L}}(X^n_{\mathcal{L}})E^n F_{\mathcal{L}}}  \rVert_1  \\
& \phantom{---} + \lVert {\rho}_{F_{\mathcal{L}}(X^n_{\mathcal{L}})E^n F_{\mathcal{L}}}  -  \bar{\rho}_{U} \otimes {\rho}_{E^nF_{\mathcal{L}}} \rVert_1 \\
& \phantom{---} + \lVert \bar{\rho}_{U} \otimes {\rho}_{E^nF_{\mathcal{L}}}  -  \bar{\rho}_{U} \otimes \widetilde{\rho}_{E^nF_{\mathcal{L}}} \rVert_1) \\
& \leq m( 2\lVert \widetilde{\rho}_{U^n_{\mathcal{L}}E^n F_{\mathcal{L}}}  -  {\rho}_{F_{\mathcal{L}}(X^n_{\mathcal{L}})E^n F_{\mathcal{L}}}  \rVert_1  \\
& \phantom{---} + \lVert {\rho}_{F_{\mathcal{L}}(X^n_{\mathcal{L}})E^n F_{\mathcal{L}}}  - \bar{\rho}_{U} \otimes {\rho}_{E^nF_{\mathcal{L}}} \rVert_1) \\
& \stackrel{(c)} \leq m( 2\mathbb{V} (\widetilde{p}_{X^n_{\mathcal{L}}U^n_{\mathcal{L}}F_{\mathcal{L}} }, {p}_{X^n_{\mathcal{L}}F_{\mathcal{L}}(X^n_{\mathcal{L}})F_{\mathcal{L}} })   \\
& \phantom{---} + \lVert {\rho}_{F_{\mathcal{L}}(X^n_{\mathcal{L}})E^n F_{\mathcal{L}}}  -  \bar{\rho}_{U} \otimes {\rho}_{E^nF_{\mathcal{L}}} \rVert_1 )\\
& \stackrel{(d)} =  m(2\mathbb{V} (p_{U^n_{\mathcal{L}}} p_{F_{\mathcal{L}}}, {p}_{F_{\mathcal{L}}(X^n_{\mathcal{L}})F_{\mathcal{L}} })   \\
& \phantom{---} + \lVert {\rho}_{F_{\mathcal{L}}(X^n_{\mathcal{L}})E^n F_{\mathcal{L}}}  -  \bar{\rho}_{U} \otimes {\rho}_{E^nF_{\mathcal{L}}} \rVert_1)\\
& \stackrel{(e)} \leq  3m  \lVert {\rho}_{F_{\mathcal{L}}(X^n_{\mathcal{L}})E^n F_{\mathcal{L}}}  -  \bar{\rho}_{U} \otimes {\rho}_{E^nF_{\mathcal{L}}} \rVert_1 \\
& \stackrel{(f)} \leq 3m \!\!\left(\!2 \cdot 2^{-n^{\xi}} \!\! +\! \sqrt{  \sum_{ {\mathcal{S} \subseteq \mathcal{L}, \mathcal{S} \neq \emptyset} } \!\!\!  2^{n [R^{\textup{U}}_{\mathcal{S}} -  H({X_{\mathcal{S}}|E})_{\rho} + \delta_{\mathcal{S}}(n) + \delta(n)]} } \right) \\
& \stackrel{(g)} \leq 3m\left(2 \cdot 2^{-n^{\xi}}  + \sqrt{  \sum_{ {\mathcal{S} \subseteq \mathcal{L}, \mathcal{S} \neq \emptyset} }   2^{-n \eta }} \right) \\
& = 3m\left(2 \cdot 2^{-n^{\xi}}  + \sqrt{  (2^{L}-1) \cdot  2^{-n \eta }} \right) \\
& \xrightarrow{ n \to \infty} 0, \numberthis \label{eqlimleakage}
\end{align*}
where $(a)$ and $(b)$ hold by the triangle inequality, $(c)$ holds by strong convexity of the trace distance and the definitions of $\widetilde{\rho}_{U^n_{\mathcal{L}} E^n F_{\mathcal{L}}} $ and  ${\rho}_{F_{\mathcal{L}}(X^n_{\mathcal{L}})E^n F_{\mathcal{L}}} $ in \eqref{eqbockE} and \eqref{eqbockE2}, $(d)$~holds by the definition of $\widetilde{p}_{X^n_{\mathcal{L}}U^n_{\mathcal{L}} F_{\mathcal{L}} }$  in \eqref{eqptilde}, $(e)$ holds because $\mathbb{V} (p_{U^n_{\mathcal{L}}} p_{F_{\mathcal{L}}}, {p}_{F_{\mathcal{L}}(X^n_{\mathcal{L}})F_{\mathcal{L}} }) \leq \lVert {\rho}_{F_{\mathcal{L}}(X^n_{\mathcal{L}})F_{\mathcal{L}}}  -  \bar{\rho}_{U} \otimes {\rho}_{F_{\mathcal{L}}} \rVert_1$, $(f)$ holds for $\xi \in ]0,1[$ by Lemma \ref{lemlhlprod} with the substitution $\epsilon \leftarrow 2^{-n^{\xi}}$ such that $\delta(n) = \log(d_E +3) \sqrt{2(\frac{1}{n}+ \frac{1}{n^{1- \xi}})}$, and $\delta_{\mathcal{S}}(n) \triangleq \log(|\mathcal{X}_{\mathcal{S}}|d_E +3) \sqrt{2 (\frac{L+1}{n}+ \frac{1}{n^{1- \xi}} )}$, $\forall \mathcal{S} \subseteq \mathcal{L}$, $(g)$ holds provided that $R^{\textup{U}}_{\mathcal{S}} \leq H(X_{\mathcal{S}}|E)_{\rho} - \delta_{\mathcal{S}}(n) - \delta(n) -\eta $, $\forall \mathcal{S} \subseteq \mathcal{L}$, $\eta>0$. 

\subsubsection{Achievable rate-tuples}
Consider the following extension of the state  described in \eqref{eqrhotildub} 
\begin{align*}
&\widetilde{\rho}_{U_{\mathcal{L}}^nX^n_{\mathcal{L}}B^nF_{\mathcal{L}}}  \triangleq  \sum_{u^n_{\mathcal{L}} \in \mathcal{U}^n_{\mathcal{L}}} \sum_{x^n_{\mathcal{L}} \in \mathcal{X}^n_{\mathcal{L}}} \sum_{f_{\mathcal{L}} \in \mathcal{F}_{\mathcal{F}}}   \widetilde{p}_{X^n_{\mathcal{L}}U^n_{\mathcal{L}}F_{\mathcal{L}}} (x^n_{\mathcal{L}},u^n_{\mathcal{L}},f_{\mathcal{L}})  \nonumber \\
&\phantom{-------} \ket{u^n_{\mathcal{L}}}\bra{u^n_{\mathcal{L}}} \otimes \ket{x^n_{\mathcal{L}}}\bra{x^n_{\mathcal{L}}} \otimes \bar{\rho}_{B^n}^{x^n_{\mathcal{L}}}\otimes  \ket{f_{\mathcal{L}}}\!\bra{f_{\mathcal{L}}}.
\end{align*}
Define also the state
\begin{align*}
&{\rho}_{U_{\mathcal{L}}^nX^n_{\mathcal{L}}B^n F_{\mathcal{L}}} \triangleq  \sum_{u^n_{\mathcal{L}} \in \mathcal{U}^n_{\mathcal{L}}} \sum_{x^n_{\mathcal{L}} \in \mathcal{X}^n_{\mathcal{L}}} \sum_{f_{\mathcal{L}} \in \mathcal{F}_{\mathcal{F}}}  {p}_{X^n_{\mathcal{L}}U^n_{\mathcal{L}}F_{\mathcal{L}}} (x^n_{\mathcal{L}},u^n_{\mathcal{L}},f_{\mathcal{L}}) \nonumber \\
&\phantom{-------} \ket{u^n_{\mathcal{L}}}\bra{u^n_{\mathcal{L}}} \otimes \ket{x^n_{\mathcal{L}}}\bra{x^n_{\mathcal{L}}} \otimes \bar{\rho}_{B^n}^{ x^n_{\mathcal{L}} }\otimes  \ket{f_{\mathcal{L}}}\!\bra{f_{\mathcal{L}}}.
\end{align*}
Then, we have 
\begin{align*}
& \max\left(\lVert\widetilde{\rho}_{X^n_{\mathcal{L}}B^n} - 
{\rho}_{X^n_{\mathcal{L}}B^n}\rVert_1 , \max_{\mathcal{S} \subseteq \mathcal{L}}\lVert\widetilde{\rho}_{U_{\mathcal{S}}^nB^n } - 
{\rho}_{U_{\mathcal{S}}^nB^n }\rVert_1 \right) \\
& \leq \lVert\widetilde{\rho}_{U_{\mathcal{L}}^nX^n_{\mathcal{L}}B^nF_{\mathcal{L}}} - 
{\rho}_{U_{\mathcal{L}}^nX^n_{\mathcal{L}}B^nF_{\mathcal{L}}}\rVert_1 \\
&\stackrel{(a)} \leq \mathbb{V} (\widetilde{p}_{X^n_{\mathcal{L}}U^n_{\mathcal{L}}F_{\mathcal{L}}} ,{p}_{X^n_{\mathcal{L}}F_{\mathcal{L}}(X^n_{\mathcal{L}})F_{\mathcal{L}}})\\
& \stackrel{(b)}= \mathbb{V} (p_{U^n_{\mathcal{L}}} p_{F_{\mathcal{L}}}, {p}_{F_{\mathcal{L}}(X^n_{\mathcal{L}})F_{\mathcal{L}} }) \\
&\xrightarrow{n\to \infty} 0
 \numberthis \label{eqlim}
\end{align*}
where $(a)$ holds by strong convexity of the trace distance, $(b)$~holds by \eqref{eqptilde}, and the limit holds by the proof of \eqref{eqlimleakage}.

Next, by Step 2 in Section \ref{secCS}, $(nR^{\textup{DC}}_l)_{l\in\mathcal{L}}$ must belong to $\mathcal{C}(\widetilde{\rho}_{U_{\mathcal{L}}^nB^n})$. One can choose  $(nR^{\textup{DC}}_l)_{l\in\mathcal{L}} \in \mathcal{C}({\rho}_{X_{\mathcal{L}}^nB^n})$ because, as proved next, we have $\mathcal{C}({\rho}_{X_{\mathcal{L}}^nB^n}) \subseteq \mathcal{C}(\widetilde{\rho}_{U_{\mathcal{L}}^nB^n})$. For $(nR^{\textup{DC}}_l)_{l\in\mathcal{L}}$ in $\mathcal{C}({\rho}_{X_{\mathcal{L}}^nB^n})$ and any $\mathcal{S} \subseteq \mathcal{L}$, we have
\begin{align*}
nR^{\textup{DC}}_{\mathcal{S}}
& \stackrel{(a)} \geq H(X_{\mathcal{S}}^n | B^n  X_{\mathcal{S}^c}^n)_{{\rho}} \\
& =  H(X_{\mathcal{L}}^n  B^n)_{{\rho}} - H( B^n  X_{\mathcal{S}^c}^n)_{{\rho}} \\
& = H(B^n | X_{\mathcal{L}}^n  )_{{\rho}} - H( B^n |  X^n_{\mathcal{S}^c})_{{\rho}}   + H(X_{\mathcal{S}}^n)_{{\rho}} \\
& \stackrel{(b)}\geq   H(B^n | X_{\mathcal{L}}^n  )_{{\rho}} - H( B^n |  U^n_{\mathcal{S}^c})_{{\rho}}   + H(X_{\mathcal{S}}^n)_{{\rho}}\\
& \stackrel{(c)}\geq   H(B^n | X_{\mathcal{L}}^n  )_{{\rho}} - H( B^n |  U^n_{\mathcal{S}^c})_{{\rho}}   + H(U_{\mathcal{S}}^n)_{{\rho}}\\
& \geq  H(B^n | X_{\mathcal{L}}^n  )_{\widetilde{\rho}} - H( B^n |  U_{\mathcal{S}^c}^n)_{\widetilde{\rho}}   + H(U_{\mathcal{S}}^n)_{\widetilde{\rho}}\\
 & \phantom{--} -|H(B^n | X_{\mathcal{L}}^n  )_{\widetilde{\rho}} -  H(B^n | X_{\mathcal{L}}^n  )_{{\rho}} | \\
 & \phantom{--} - | H( B^n |  U^n_{\mathcal{S}^c})_{\widetilde{\rho}} - H( B^n |  U^n_{\mathcal{S}^c})_{{\rho}} | \\
 & \phantom{--} - | H(U_{\mathcal{S}}^n)_{\widetilde{\rho}} - H(U_{\mathcal{S}}^n)_{{\rho}}|\\
& \stackrel{(d)}\geq   H(B^n | X_{\mathcal{L}}^n  )_{\widetilde{\rho}} - H( B^n |  U_{\mathcal{S}^c}^n)_{\widetilde{\rho}}   + H(U_{\mathcal{S}}^n)_{\widetilde{\rho}} - o(n) \\
&  \stackrel{(e)}\geq  H(B^n | U_{\mathcal{L}}^n  )_{\widetilde{\rho}} - H( B^n |  U_{\mathcal{S}^c}^n)_{\widetilde{\rho}}   + H(U_{\mathcal{S}}^n)_{\widetilde{\rho}} - o(n) \\
& =  H( U_{\mathcal{S}}^n  |B^n   U_{\mathcal{S}^c}^n)_{\widetilde{\rho}}   - o(n),
\end{align*}
where $(a)$ holds because $(nR^{\textup{DC}}_l)_{l\in\mathcal{L}}$ in $\mathcal{C}({\rho}_{X_{\mathcal{L}}^nB^n})$, $(b)$ holds by the quantum data processing inequality because, by definition of ${\rho}$, for any $\mathcal{S} \subseteq \mathcal{L}$, $U_{\mathcal{S}}^n$ is a function of $X_{\mathcal{S}}^n$, $(c)$ holds by Lemma \ref{lemlhlprod} because, by definition of $\rho$, for any $\mathcal{S} \subseteq \mathcal{L}$, $U_{\mathcal{S}}^n$ is the output of hash functions when $X_{\mathcal{S}}^n$ is the input, $(d)$~holds by the Alicki-Fannes inequality and \eqref{eqlim}, $(e)$ holds by the quantum data processing inequality because, by definition of~$\widetilde{\rho}$, $\widetilde{X}_{\mathcal{L}}^n$ is a function of $U_{\mathcal{L}}^n$.

Hence, by having chosen $(nR^{\textup{DC}}_l)_{l\in\mathcal{L}} \in \mathcal{C}({\rho}_{X_{\mathcal{L}}^nB^n})$ and the choice of $(R_l^{\textup{U}})_{l\in\mathcal{L}}$ in \eqref{eqlimleakage}, we have the system
\begin{align}
\begin{pmatrix}
	R^{\textup{DC}}_{\mathcal{S}} \geq H(X_{\mathcal{S}} | B  X_{\mathcal{S}^c})_{{\rho}}, \forall \mathcal{S} \subseteq \mathcal{L} \\
	R^{\textup{U}}_{\mathcal{S}} \leq H(X_{\mathcal{S}}|E)_{\rho} , \forall \mathcal{S} \subseteq \mathcal{L}
	\end{pmatrix} ,
\end{align}
which we rewrite, by Step 3 in Section \ref{secCS}, as 
\begin{align}
\begin{pmatrix}
	R^{\textup{DC}}_{\mathcal{S}} \geq H(X_{\mathcal{S}} | B  X_{\mathcal{S}^c})_{{\rho}}, \forall \mathcal{S} \subseteq \mathcal{L} \\
	R_{\mathcal{S}} + R^{\textup{DC}}_{\mathcal{S}} \leq H(X_{\mathcal{S}}|E)_{\rho} , \forall \mathcal{S} \subseteq \mathcal{L}
	\end{pmatrix} \label{eqSI}.
\end{align}
Next, by  Lemma \ref{submodular2}, the set functions $\mathcal{S} \mapsto - H(X_{\mathcal{S}} | B  X_{\mathcal{S}^c})_{{\rho}}$ and $\mathcal{S} \mapsto  H(X_{\mathcal{S}} | E)_{{\rho}}- R_{\mathcal{S}}$ are submodular. Hence, by Lemma~\ref{lemsubm}, the system \eqref{eqSI} has a solution if and only if 
\begin{align}
H(X_{\mathcal{S}} | B X_{\mathcal{S}^c})_{{\rho}} \leq 	 H(X_{\mathcal{S}}|E)_{\rho} - R_{\mathcal{S}} , \forall \mathcal{S} \subseteq \mathcal{L},
\end{align}
which we rewrite as 
\begin{align*}
	R_{\mathcal{S}}  
	& \leq H(X_{\mathcal{S}}|E)_{\rho} -H(X_{\mathcal{S}} | B X_{\mathcal{S}^c})_{{\rho}} \\
	& = I(X_{\mathcal{S}};B |X_{\mathcal{S}^c})_{\rho} - I(X_{\mathcal{S}};E)_{\rho} , \forall \mathcal{S} \subseteq \mathcal{L}.
\end{align*}

\subsubsection{Expurgation}
We write the average probability of error and average message indistinguishability of the coding scheme in Section \ref{secCS} as $\mathbf{S}_n \triangleq \lVert \widetilde{\rho}_{M^m_{\mathcal{L}}E^{mn} F^m_{\mathcal{L}}}  -  \rho_U^{\otimes m} \otimes \widetilde{\rho}_{E^{mn}F^m_{\mathcal{L}}} \rVert_1$ and $\mathbf{P}_n \triangleq \mathbb{P} \left[ \psi (\bar{\rho}_{B^{mn}}^{\phi_{\mathcal{L}}(M^m_{\mathcal{L}})} ) \neq M^m_{\mathcal{L}}\right] $, respectively. To simplify notation, we write $\mathbf{m}_{\mathcal{L}} \triangleq m^m_{\mathcal{L}}$ for $m^m_{\mathcal{L}} \in \mathcal{M}^m_{\mathcal{L}}$. Then, we have
\begin{align*}
	\mathbf{S}_n  &= \sum_{\mathbf{m}_{\mathcal{L}}} \frac{1}{|\mathcal{M}^m_{\mathcal{L}}|}S_n(\mathbf{m}_{\mathcal{L}}), \\
	\mathbf{P}_n  &= \sum_{\mathbf{m}_{\mathcal{L}}} \frac{1}{|\mathcal{M}^m_{\mathcal{L}}|} P_n(\mathbf{m}_{\mathcal{L}}),
\end{align*}
where for $\mathbf{m}_{\mathcal{L}} \in \mathcal{M}^m_{\mathcal{L}}$, we have defined
\begin{align*}
S_n(\mathbf{m}_{\mathcal{L}}) &\triangleq  \lVert \widetilde{\rho}^{\mathbf{m}_{\mathcal{L}}}_{E^{nm} F^m_{\mathcal{L}}}  -   \widetilde{\rho}_{E^{nm}F^m_{\mathcal{L}}} \rVert_1 ,\\
P_n(\mathbf{m}_{\mathcal{L}}) &\triangleq \mathbb{P} \left[ \psi (\bar{\rho}_{B^{mn}}^{\phi_{\mathcal{L}}(\mathbf{M}_{\mathcal{L}})} ) \neq \mathbf{M}_{\mathcal{L}} | \mathbf{M}_{\mathcal{L}} = \mathbf{m}_{\mathcal{L}}\right].
\end{align*} 
 Let $\alpha \in ]0,1[$. By Markov's inequality and \eqref{eqproofrelia}, \eqref{eqlimleakage}, for at least a fraction $1-\alpha$ of the codewords, $P_n(\mathbf{m}_{\mathcal{L}}) \leq  \alpha^{-1} \mathbf{P}_n$ and for at least a fraction $1-\alpha$ of the codewords, $S_n(\mathbf{m}_{\mathcal{L}}) \leq \alpha^{-1} \mathbf{S}_n $. Hence, for a fraction of the codewords at least $1-2 \alpha$, $ P_n(\mathbf{m}_{\mathcal{L}})\leq \alpha^{-1} \mathbf{P}_n \xrightarrow{n \to \infty} 0$ and $ S_n(\mathbf{m}_{\mathcal{L}})\leq \alpha^{-1} \mathbf{S}_n \xrightarrow{n \to \infty} 0$. Finally, we expurgate the code to only retain this fraction $1-2\alpha$ of messages, which has a negligible impact on the asymptotic communication rates.

\section{Converse of Theorem \ref{th1}} \label{secconv}

Similar to the case of point-to-point channels, e.g., \cite[Sec. 23.4]{wilde2013quantum}, it is sufficient to consider the task of exchanging private randomness  between the transmitters and the legitimate receiver, which is a weaker task than private classical communication. Specifically, assume that Transmitter~$l\in\mathcal{L}$ prepares a maximally correlated state $\rho_{M_l M_l'}$ and encodes $M'_l$ as $\rho^{m_l}_{A^{\prime n}_l}$, $m_l \in \mathcal{M}_l$, such that the legitimate receiver can recover the share $M'_{\mathcal{L}}$ of the state $\rho_{M_{\mathcal{L}}M_{\mathcal{L}}'} \triangleq \bigotimes_{l \in \mathcal{L}} \rho_{M_l M_l'}$ with some decoder $\mathcal{D}_{B^n \to M'_{\mathcal{L}}}$. The state resulting from this encoding and $n$ independent uses of the channel, i.e.,  $\mathcal{N}_{A^{\prime n}_{\mathcal{L}}\to B^n}$, is 
$$
\omega_{M_{\mathcal{L}}B^nE^n} \triangleq \frac{1}{|\mathcal{M}_{\mathcal{L}}|} \sum_{m_{\mathcal{L}} \in \mathcal{M}_{\mathcal{L}}} \ket{m_{\mathcal{L}}} \bra{m_{\mathcal{L}}} \otimes \mathcal{U}_{A^{\prime n}_{\mathcal{L}} \to B^nE^n}^{\mathcal{N}} (\rho^{m_{\mathcal{L}}}_{A^{\prime n}_{\mathcal{L}}}),
$$
where $\rho^{m_{\mathcal{L}}}_{A^{\prime n}_{\mathcal{L}}} \triangleq \bigotimes_{l \in \mathcal{L}}\rho^{m_l}_{A^{\prime n}_l}$ and $\ket{m_{\mathcal{L}}} \bra{m_{\mathcal{L}}} \triangleq \bigotimes_{l \in \mathcal{L}} \ket{m_{l}} \bra{m_{l}}$, with $m_{\mathcal{L}} = (m_l)_{l\in\mathcal{L}} \in \mathcal{M}_{\mathcal{L}}$. Then, the decoder of the legitimate receiver produces
$$
\omega_{M_{\mathcal{L}}M_{\mathcal{L}}'E^n} \triangleq \mathcal{D}_{B^n \to M'_{\mathcal{L}}}(\omega_{M_{\mathcal{L}}B^nE^n} ),
$$
and privacy with respect to the environment is assumed, i.e., there exists a constant state $\sigma_{E^n}$ independent of $\rho_{M_{\mathcal{L}}M_{\mathcal{L}}'}$ such~that
\begin{align} \label{eqapprox}
 \lVert \omega_{M_{\mathcal{L}}M_{\mathcal{L}}'E^n} - \rho_{M_{\mathcal{L}} M_{\mathcal{L}}'} \otimes \sigma_{E^n} \rVert_1 \leq \delta(n),
\end{align}
where $ \lim_{n\to \infty} \delta(n) = 0$.
Next, for $\mathcal{S} \subseteq \mathcal{L}$, we have
\begin{align*}
{nR_{\mathcal{S}}} 
& = \sum_{l \in \mathcal{S}} \log |\mathcal{M}_l| \\
& = \sum_{l \in \mathcal{S}} I(M_{l};M_{l}')_{\rho}  \\
& \stackrel{(a)} = I(M_{\mathcal{S}};M_{\mathcal{S}}')_{\rho} \\
& = H(M_{\mathcal{S}})_{\rho} - H(M_{\mathcal{S}}|M_{\mathcal{S}}')_{\rho} \\
& \stackrel{(b)} = H(M_{\mathcal{S}} | M_{\mathcal{S}^c})_{\rho} - H(M_{\mathcal{S}}|M_{\mathcal{S}}')_{\rho} \\ 
& \stackrel{(c)}\leq H(M_{\mathcal{S}} | M_{\mathcal{S}^c})_{\rho} - H(M_{\mathcal{S}}|M_{\mathcal{S}}'M_{\mathcal{S}^c})_{\rho} \\ 
& \leq H(M_{\mathcal{S}} | M_{\mathcal{S}^c})_{\omega} - H(M_{\mathcal{S}}|M_{\mathcal{S}}'M_{\mathcal{S}^c})_{\omega} \\ 
& \phantom{--}  + |H(M_{\mathcal{S}} | M_{\mathcal{S}^c})_{\omega} - H(M_{\mathcal{S}} | M_{\mathcal{S}^c})_{\rho} | \\ 
& \phantom{--}  + |H(M_{\mathcal{S}}|M_{\mathcal{S}}'M_{\mathcal{S}^c})_{\omega} - H(M_{\mathcal{S}}|M_{\mathcal{S}}'M_{\mathcal{S}^c})_{\rho}|\\
&  \stackrel{(d)}\leq H(M_{\mathcal{S}} | M_{\mathcal{S}^c})_{\omega} - H(M_{\mathcal{S}}|M_{\mathcal{S}}'M_{\mathcal{S}^c})_{\omega} + o(n) \\
& = I(M_{\mathcal{S}};M_{\mathcal{S}}' | M_{\mathcal{S}^c})_{\omega} + o(n) \\
& \stackrel{(e)}\leq  I(M_{\mathcal{S}};B^n |M_{\mathcal{S}^c})_{\omega} + o(n)\\
& \stackrel{(f)}\leq  I(M_{\mathcal{S}};B^n |M_{\mathcal{S}^c})_{\omega} - I(M_{\mathcal{S}};E^n )_{\omega}  + o(n), \numberthis \label{eqconv}
\end{align*}
where $(a)$ holds because $\rho_{M_{\mathcal{S}}M_{\mathcal{S}}'} = \bigotimes_{l \in \mathcal{S}} \rho_{M_l M_l'}$, $(b)$ holds because for any $\mathcal{S}, \mathcal{T} \subseteq \mathcal{L}$ such that $\mathcal{S} \cap \mathcal{T} = \emptyset$, we have $\rho_{M_{\mathcal{S}} M_{\mathcal{T}}}  = \rho_{M_{\mathcal{S}}} \otimes \rho_{ M_{\mathcal{T}}}$, $(c)$ holds because conditioning does not increase entropy, $(d)$ holds by \eqref{eqapprox} and Alicki-Fannes inequality, $(e)$ holds by the quantum data processing inequality, $(f)$ holds because $I(M_{\mathcal{S}};E^n )_{\omega} = H(M_{\mathcal{S}}|E^n)_{\rho \otimes \sigma} - H(M_{\mathcal{S}}|E^n )_{\omega}$ is upper bounded by $o(n)$ using  Alicki-Fannes inequality and \eqref{eqapprox}. Finally, from \eqref{eqconv} we conclude that $(R_l)_{l\in\mathcal{L}}$ belongs to $ \textup{cl} \left(  \bigcup_{n=1}^{\infty} \frac{1}{n} \mathcal{P} (\mathcal{N}^{\otimes n}) \right)$.

\section{Proof of Theorems \ref{th2} and \ref{th3}}  \label{secsumpro}

We first prove the following lemma, which provides a regularized expression of the best achievable sum-rate in $C_{\textup{P-MAC}}$  for degradable channels.

\begin{lem} \label{lemsigmaP}
Let $\mathcal{N}$ be a degradable quantum multiple-access channel. We have
\begin{align} \label{eqsumratreg}
 C_{\textup{P-MAC}}^{\textup{sum}}  (\mathcal{N}) = \lim_{n\to \infty} \frac{1}{n} P_{\textup{MAC}}^{\textup{sum}} (\mathcal{N}^{\otimes n}),
\end{align}
where $P_{\textup{MAC}}^{\textup{sum}}$ is defined in \eqref{eqPmacsum}.
\end{lem}

\begin{proof}
Note that by Theorem \ref{th1} the inequality $C_{\textup{P-MAC}}^{\textup{sum}}  (\mathcal{N}) \leq \lim_{n\to \infty} \frac{1}{n} P_{\textup{MAC}}^{\textup{sum}} (\mathcal{N}^{\otimes n})$ is trivial. It is thus sufficient to show the achievability of the sum-rate $\lim_{n\to \infty} \frac{1}{n} P_{\textup{MAC}}^{\textup{sum}} (\mathcal{N}^{\otimes n})$. Consider the set function $f_{\rho}:\mathcal{S} \mapsto I(X_{\mathcal{S}};B |X_{\mathcal{S}^c})_{\rho} - I(X_{\mathcal{S}};E)_{\rho}$, where $\rho$ is a state as defined in Theorem \ref{th1}. By Lemma~\ref{submodular2} in Appendix \ref{App_6}, $f_{\rho}$ is submodular. Next,  $f_{\rho}$ is also non-negative because for any $\mathcal{S} \subseteq \mathcal{L}$
 \begin{align*}
 f_{\rho} (\mathcal{S} ) 
 & = I(X_{\mathcal{S}};B |X_{\mathcal{S}^c})_{\rho} - I(X_{\mathcal{S}};E)_{\rho}\\
 & \stackrel{(a)}= I(X_{\mathcal{S}};B X_{\mathcal{S}^c})_{\rho} - I(X_{\mathcal{S}};E)_{\rho}\\
 & \stackrel{(b)}\geq  I(X_{\mathcal{S}};B )_{\rho} - I(X_{\mathcal{S}};E)_{\rho}\\
 & \stackrel{(c)}\geq 0,
 \end{align*}
 where $(a)$ holds because for any $\mathcal{S} \subseteq \mathcal{L}$, we have $\rho_{X_{\mathcal{S}} X_{\mathcal{S}^c}}  = \rho_{X_{\mathcal{S}}} \otimes \rho_{ X_{\mathcal{S}^c}}$, $(b)$ holds by the chain rule and positivity of mutual information, $(c)$ holds by the quantum data processing inequality because $\mathcal{N}$ is degradable.

Hence, $f_{\rho}$ is submodular and non-negative. However, $f_{\rho}$ is not necessarily non-decreasing, which means that $\mathcal{R} (f_{\rho}) \triangleq\{ (R_l)_{l\in \mathcal{L}}:R_{\mathcal{S}} \leq f_{\rho}(\mathcal{S}) , \forall \mathcal{S} \subseteq \mathcal{L}  \}$ associated with the function $f_{\rho}$ does not describe a polymatroid in general -- see Definition \ref{defpoly} in Appendix \ref{App_6}. To overcome this difficulty, we define the set function $f_{\rho}^{*}$ with
$$
f_{\rho}^{*}:\mathcal{S} \mapsto   \min_{ \substack{\mathcal{A} \subseteq \mathcal{L} \\ \text{s.t. } \mathcal{A} \supseteq \mathcal{S}}} f_{\rho}(\mathcal{A}).
$$
By Lemma \ref{propolyplus} in Appendix \ref{App_6}, the set function $f_{\rho}^{*}$ is normalized, i.e., $f_{\rho}^{*}(\emptyset)=0$, non-decreasing, and submodular because $f_{\rho}$ is normalized, non-negative, and submodular. Hence, $\mathcal{R} (f_{\rho}^{*})$ associated with the function $f_{\rho}^{*}$ describes a polymatroid and by \cite{edmonds2003submodular} its dominant face, i.e., $\{ (R_l)_{l\in \mathcal{L}} \in \mathcal{R} (f^{*}_{\rho}):R_{\mathcal{L}} = f^{*}_{\rho}(\mathcal{L})   \}$ is non-empty. Consequently, there exists a rate-tuple $(R_l)_{l\in \mathcal{L}} \in \mathcal{R} (f_{\rho}^{*})$ such that $R_{\mathcal{L}} = f_{\rho}^{*}(\mathcal{L})$. Next, by inspecting $\mathcal{R} (f_{\rho}^{*})$ and $\mathcal{R} (f_{\rho})$, we have that  $\mathcal{R} (f_{\rho}^{*})=\mathcal{R} (f_{\rho})$ by the construction of $f_{\rho}^{*}$. We also have $f_{\rho}^{*}(\mathcal{L}) = f_{\rho}(\mathcal{L})$ by the construction of~$f_{\rho}^{*}$. Hence, we conclude that there exists a rate-tuple $(R_l)_{l\in \mathcal{L}} \in \mathcal{R} (f_{\rho})$ such that $R_{\mathcal{L}} = f_{\rho}(\mathcal{L})$. Finally, from Theorem \ref{th1}, we conclude that the sum-rate $\lim_{n\to \infty} \frac{1}{n} P_{\textup{MAC}}^{\textup{sum}} (\mathcal{N}^{\otimes n})$ is achievable, and thus that~\eqref{eqsumratreg}~holds.
\end{proof}

Next, we prove  the following equality.
\begin{lem}  \label{lemsigmaP2}
Let $\mathcal{N}$ be a degradable quantum multiple-access channel. We have $$P_{\textup{MAC}}^{\textup{sum}}  (\mathcal{N}) = Q^{\textup{sum}}_{\textup{MAC}} (\mathcal{N} ).$$
\end{lem}
\begin{proof}
See Appendix \ref{App_8}.
\end{proof}
Finally, we have that $Q^{\textup{sum}}_{\textup{MAC}}$ is additive for degradable channels. The proof of Lemma \ref{lemsigmaP3} is similar to the proof of additivity for the coherent information of degradable channels. Note that  Lemma \ref{lemsigmaP3} is also referenced in \cite{yard2008capacity}.
\begin{lem}  \label{lemsigmaP3}
Let $\mathcal{N}$ and $\mathcal{M}$ be two degradable quantum multiple-access channels. Then, we have $$Q^{\textup{sum}}_{\textup{MAC}} (\mathcal{N} \otimes \mathcal{M}) = Q^{\textup{sum}}_{\textup{MAC}} (\mathcal{N} ) + Q^{\textup{sum}}_{\textup{MAC}} ( \mathcal{M}).$$
\end{lem}
\noindent{} All in all, from Lemmas \ref{lemsigmaP}, \ref{lemsigmaP2}, \ref{lemsigmaP3}, we obtain Theorems \ref{th2} and~\ref{th3}.

\section{Concluding remarks} \label{secconcl}
We introduced the notion of private capacity region for quantum multiple-access channels. For an arbitrary number of transmitters, we derived a regularized expression for this private capacity region. In the case of degradable channels, we also derived two single-letter expressions for the best achievable  sum-rate. One of these expressions coincides with the best achievable sum-rate for quantum communication over degradable quantum multiple-access channels. 

Our proof technique for the achievability part relies on an emulation of a proof based on random binning. Specifically, our achievability result decouples the reliability and privacy constraints, which are handled via distributed source coding with quantum side information at the receiver and distributed hashing, respectively. Consequently, we reduced a multiuser coding problem into multiple single-user coding problems. Indeed, distributed source coding with quantum side information at the receiver can be reduced to single-user source coding with quantum side information at the receiver, and distributed hashing is, by construction, performed independently at each~transmitter. 

As part of our proof, we derived a distributed leftover hash lemma in the presence of quantum side information, which may be of independent interest. Note that in our setting the seeds size needed to choose the hash functions is irrelevant. However, for other applications, it may be desirable to reduce the necessary seeds size. Specifically, it remains open to extend our result to $\delta$-almost two-universal hash functions, which are known to enable a reduction of the necessary seed size for the non-distributed setting, i.e., the special case $L=1$,~\cite{tomamichel2011leftover}.

\appendices

\section{Proof of Lemma \ref{lemlhl}} \label{app_1}
For any $\rho_{XE} \in \mathcal{S}_{\leq}( \mathcal{H}_X \otimes \mathcal{H}_E)$ and $\sigma_E \in  \mathcal{S}_{=}( \mathcal{H}_E)$, the collision entropy of $\rho_{XE}$ relative to $\sigma_E$~\cite{renner2008security} is defined as 
\begin{align}
H_2(\rho_{XE}|\sigma_E) \triangleq - \log \frac{ \Tr[ \left(\rho_{XE}( I_X \otimes \sigma_E^{-1/2}) \right)^2 ]}{\Tr \rho_{XE}}. \label{defcol}
\end{align}
Next, define $A_{\mathcal{L}} \triangleq F_{\mathcal{L}}(X_{\mathcal{L}})$. We then have
\begin{align}
&\lVert \rho_{A_{\mathcal{L}} E F_{\mathcal{L}}}  -  \rho_U \otimes \rho_{EF_{\mathcal{L}}} \rVert_1  \nonumber \\ \nonumber
& \stackrel{(a)}{=}  \mathbb{E}_{F_{\mathcal{L}}}    \left\lVert \rho_{A_{\mathcal{L}}E}^{F_{\mathcal{L}}}  -  \rho_U \otimes \rho_{E} \right\rVert_1 \\ \nonumber
& \stackrel{(b)}{\leq} \mathbb{E}_{F_{\mathcal{L}}}  \sqrt{2^{r_{\mathcal{L}}} } \sqrt{\Tr [ \left( (\rho^{F_{\mathcal{L}}}_{A_{\mathcal{L}}E}  -  \rho_U \otimes \rho_{E})(I_{A_{\mathcal{L}}} \otimes \sigma_E^{-1/2} ) \right)^2 ]} \\ \nonumber
& \stackrel{(c)}{\leq}  \sqrt{2^{r_{\mathcal{L}}} } \sqrt{\mathbb{E}_{F_{\mathcal{L}}} \Tr [ \left( (\rho^{F_{\mathcal{L}}}_{A_{\mathcal{L}}E}  -  \rho_U \otimes \rho_{E})(I_{A_{\mathcal{L}}} \otimes \sigma_E^{-1/2} ) \right)^2 ]} \\ \nonumber
& \stackrel{(d)}{=} \sqrt{2^{r_{\mathcal{L}}} }  \left( \mathbb{E}_{F_{\mathcal{L}}}   \Tr \left[ \left(  \sum_{a_{\mathcal{L}} \in \mathcal{A}_{\mathcal{L}}} \ket{a_{\mathcal{L}}}\!\bra{a_{\mathcal{L}}} \right. \right. \right.\\ \nonumber
& \phantom{--} \smash{ \left.\left.  \left. \otimes \! \left(\sigma_E^{-1/4}  {\rho}^{F_{\mathcal{L}},a_{\mathcal{L}}}_{E} \sigma_E^{-1/4} \!-\! 2^{-r_{\mathcal{L}}} \sigma_E^{-1/4} \rho_{E}\sigma_E^{-1/4} \right) \right)^2 \right] \right)^{1/2} } \\  \nonumber \\ \nonumber
& = \sqrt{2^{r_{\mathcal{L}}} } \left( \mathbb{E}_{F_{\mathcal{L}}}    \smash{\sum_{a_{\mathcal{L}} \in \mathcal{A}_{\mathcal{L}}} } \Tr \left[    \left(\sigma_E^{-1/4}  {\rho}^{F_{\mathcal{L}},a_{\mathcal{L}}}_{E} \sigma_E^{-1/4} \right. \right. \right.  \\ \nonumber
&\phantom{-------ll----}\smash{ \left.\left.\left.  - 2^{-r_{\mathcal{L}}} \sigma_E^{-1/4} \rho_{E}\sigma_E^{-1/4}  \right)^2 \right] \right)^{1/2}} \\ \nonumber \\
&\stackrel{(e)} =  \sqrt{2^{r_{\mathcal{L}}} } \smash{\left( \mathbb{E}_{F_{\mathcal{L}}}   \sum_{a_{\mathcal{L}}\in \mathcal{A}_{\mathcal{L}}}  \Tr \left[ \left( \sigma_E^{-1/4} {\rho}^{F_{\mathcal{L}},a_{\mathcal{L}}}_{E} \sigma_E^{-1/4}  \right)^2 \right]   \right. } \nonumber \\
& \phantom{--------ll-} \left. - 2^{-r_{\mathcal{L}}}  \Tr \left[ \left( \sigma_E^{-1/4} \rho_{E} \sigma_E^{-1/4}  \right)^2 \right] \right)^{1/2}\!\!\!\!\!\!\!\!, \label{eq1}
\end{align}
where $(a)$ holds with $\rho^{F_{\mathcal{L}}}_{A_{\mathcal{L}}E} \triangleq \sum_{a_{\mathcal{L}} \in \mathcal{A}_{\mathcal{L}} } \ket{a_{\mathcal{L}}}\!\bra{a_{\mathcal{L}}} \otimes {\rho}_E^{F_{\mathcal{L}},a_{\mathcal{L}}}$, $(b)$ holds by Lemma \ref{lem_app} in Appendix~\ref{App_6} with $\rho \triangleq \rho^{F_{\mathcal{L}}}_{A_{\mathcal{L}}E}  -  \rho_U \otimes \rho_{E}$ and $\sigma \triangleq I_{A_{\mathcal{L}}} \otimes \sigma_E$ for any $\sigma_E \in \mathcal{S}_{\leq}(\mathcal{H}_E)$, $(c)$ holds by Jensen's inequality, $(d)$ holds because
\begin{align*}
& \Tr [\left( (\rho^{F_{\mathcal{L}}}_{A_{\mathcal{L}}E}  -  \rho_U \otimes \rho_{E})(I_{A_{\mathcal{L}}} \otimes \sigma_E^{-1/2} ) \right)^2 ]\\
& = \Tr \left[ \left( (I_{A_{\mathcal{L}}} \otimes \sigma_E^{-1/4} ) \right. \right. \\
&\left. \left. \phantom{l} \cdot \! \left[ \sum_{a_{\mathcal{L}} \in \mathcal{A}_{\mathcal{L}}} \!\! \ket{a_{\mathcal{L}}}\!\bra{a_{\mathcal{L}}} \otimes \left( {\rho}^{F_{\mathcal{L}},a_{\mathcal{L}}}_{E} \!- \! 2^{-r_{\mathcal{L}}} \rho_{E}\right)\! \right]\! \! (I_{A_{\mathcal{L}}} \!\otimes\! \sigma_E^{-1/4} )\!\!\right)^{\!\!2} \right]\!\!,
\end{align*}
$(e)$ holds by expanding and simplifying the square inside the trace. Next, we have
\begin{align}
& \sum_{a_{\mathcal{L}}\in \mathcal{A}_{\mathcal{L}}} \Tr [ \left( \sigma_E^{-1/4} {\rho}^{F_{\mathcal{L}},a_{\mathcal{L}}}_{E} \sigma_E^{-1/4}  \right)^2 ] \nonumber \\ \nonumber
& =  \sum_{a_{\mathcal{L}}\in \mathcal{A}_{\mathcal{L}}} \Tr \left[  \sigma_E^{-1/4} \left( \sum_{x_{\mathcal{L}}\in F^{-1}_{\mathcal{L}}(a_{\mathcal{L}})}{\rho}^{x_{\mathcal{L}}}_{E}\right) \sigma_E^{-1/2} \right. \nonumber \\ \nonumber 
& \phantom{-------------} \left. \cdot \left(\sum_{x'_{\mathcal{L}}\in F^{-1}_{\mathcal{L}}(a_{\mathcal{L}})}{\rho}^{x'_{\mathcal{L}}}_{E} \right) \sigma_E^{-1/4}   \right]\\ \nonumber
& =  \sum_{a_{\mathcal{L}}\in \mathcal{A}_{\mathcal{L}}}  \sum_{x_{\mathcal{L}},x'_{\mathcal{L}}\in F^{-1}_{\mathcal{L}}(a_{\mathcal{L}})}\Tr [  \sigma_E^{-1/4}  {\rho}^{x_{\mathcal{L}}}_{E} \sigma_E^{-1/2}  {\rho}^{x'_{\mathcal{L}}}_{E} \sigma_E^{-1/4}   ]\\ \nonumber
&\stackrel{(a)} =  \sum_{a_{\mathcal{L}}\in \mathcal{A}_{\mathcal{L}}} \sum_{\mathcal{S} \subseteq \mathcal{L}}  \sum_{x_{\mathcal{L}}\in F^{-1}_{\mathcal{L}}(a_{\mathcal{L}})} \sum_{\substack{x'_{\mathcal{L}}\in F^{-1}_{\mathcal{L}}(a_{\mathcal{L}})\\ \textup{s.t } x'_{\mathcal{S}} \neq   x_{\mathcal{S}} \\ \phantom{\textup{s.t }} x'_{\mathcal{S}^c} =   x_{\mathcal{S}^c}}} 1 \nonumber \\ \nonumber 
& \phantom{---}\times \Tr [  \sigma_E^{-1/4}  {\rho}^{x_{\mathcal{L}}}_{E} \sigma_E^{-1/2}  {\rho}^{x'_{\mathcal{L}}}_{E} \sigma_E^{-1/4}   ]\\ \nonumber
&  =  \sum_{a_{\mathcal{L}}\in \mathcal{A}_{\mathcal{L}}} \sum_{\mathcal{S} \subseteq \mathcal{L}}  \sum_{x_{\mathcal{L}}\in \mathcal{X}_{\mathcal{L}}} \!\!\! \sum_{\substack{x'_{\mathcal{L}}\in \mathcal{X}_{\mathcal{L}}\\ \textup{s.t } x'_{\mathcal{S}} \neq   x_{\mathcal{S}} \\ \phantom{\textup{s.t }} x'_{\mathcal{S}^c} =   x_{\mathcal{S}^c}}} \!\!\!\!\!\! \mathds{1}\{ x_{\mathcal{L}}\in F^{-1}_{\mathcal{L}}(a_{\mathcal{L}}) \} \nonumber \\ \nonumber 
& \phantom{---}\times \mathds{1}\{ x'_{\mathcal{S}}\in F^{-1}_{\mathcal{S}}(a_{\mathcal{S}}) \}  \Tr [  \sigma_E^{-1/4}  {\rho}^{x_{\mathcal{L}}}_{E} \sigma_E^{-1/2}  {\rho}^{x'_{\mathcal{L}}}_{E} \sigma_E^{-1/4}   ]\\ \nonumber
& \stackrel{(b)} =   \sum_{\mathcal{S} \subseteq \mathcal{L}} \sum_{a_{\mathcal{S}}\in \mathcal{A}_{\mathcal{S}}}   \sum_{x_{\mathcal{L}}\in \mathcal{X}_{\mathcal{L}}} \!\!\! \sum_{\substack{x'_{\mathcal{L}}\in \mathcal{X}_{\mathcal{L}}\\ \textup{s.t } x'_{\mathcal{S}} \neq   x_{\mathcal{S}} \\ \phantom{\textup{s.t }} x'_{\mathcal{S}^c} =   x_{\mathcal{S}^c}}} \!\!\!\!\!\! \mathds{1}\{ x_{\mathcal{S}}, x'_{\mathcal{S}}\in F^{-1}_{\mathcal{S}}(a_{\mathcal{S}}) \}   \nonumber \\ \nonumber 
& \phantom{---}\times \Tr [  \sigma_E^{-1/4}  {\rho}^{x_{\mathcal{L}}}_{E} \sigma_E^{-1/2}  {\rho}^{x'_{\mathcal{L}}}_{E} \sigma_E^{-1/4}   ]\\
& \stackrel{(c)}=  \sum_{\mathcal{S} \subseteq \mathcal{L}}  \sum_{x_{\mathcal{L}}\in \mathcal{X}_{\mathcal{L}}} \!\!\! \sum_{\substack{x'_{\mathcal{L}}\in \mathcal{X}_{\mathcal{L}}\\ \textup{s.t } x'_{\mathcal{S}} \neq   x_{\mathcal{S}} \\ \phantom{\textup{s.t }} x'_{\mathcal{S}^c} =   x_{\mathcal{S}^c}}} \!\!\!\!\!\! \mathds{1}\{ F_{\mathcal{S}}(x_{\mathcal{S}} ) = F_{\mathcal{S}}(x'_{\mathcal{S}} ) \}  \nonumber \\
& \phantom{---}\times \Tr [  \sigma_E^{-1/4}  {\rho}^{x_{\mathcal{L}}}_{E} \sigma_E^{-1/2}  {\rho}^{x'_{\mathcal{L}}}_{E} \sigma_E^{-1/4}   ], \label{eq2}
\end{align}
where in $(a)$ the notation $x'_{\mathcal{S}} \neq   x_{\mathcal{S}}$ means  $\forall l\in\mathcal{S}$, $x'_{l} \neq   x_{l}$ , $(b)$ holds because $ \sum_{a_{\mathcal{S}^c}\in \mathcal{A}_{\mathcal{S}^c}}\mathds{1}\{ x_{\mathcal{S}^c}\in F^{-1}_{\mathcal{L}}(a_{\mathcal{S}^c}) \} = 1$, and $(c)$ holds because $\sum_{a_{\mathcal{S}}\in \mathcal{A}_{\mathcal{S}}} \mathds{1}\{ x_{\mathcal{S}}, x'_{\mathcal{S}}\in F^{-1}_{\mathcal{S}}(a_{\mathcal{S}}) \}  = \mathds{1}\{ F_{\mathcal{S}}(x_{\mathcal{S}} ) = F_{\mathcal{S}}(x'_{\mathcal{S}} ) \} $. Then, taking the expectation over ${F_{\mathcal{L}}}$ in \eqref{eq2}, we obtain
\begin{align}
& \mathbb{E}_{F_{\mathcal{L}}}  \sum_{a_{\mathcal{L}}\in \mathcal{A}_{\mathcal{L}}} \Tr [\left( \sigma_E^{-1/4} {\rho}^{F_{\mathcal{L}},a_{\mathcal{L}}}_{E} \sigma_E^{-1/4}  \right)^2 ] \nonumber \\ \nonumber
& =  \smash{\sum_{\mathcal{S} \subseteq \mathcal{L}}  \sum_{x_{\mathcal{L}}\in \mathcal{X}_{\mathcal{L}}} \!\!\! \sum_{\substack{x'_{\mathcal{L}}\in \mathcal{X}_{\mathcal{L}}\\ \textup{s.t } x'_{\mathcal{S}} \neq   x_{\mathcal{S}} \\ \phantom{\textup{s.t }} x'_{\mathcal{S}^c} =   x_{\mathcal{S}^c}}} }\!\!\!\!\!\! \mathbb{E}_{F_{\mathcal{L}}}  \mathds{1}\{ F_{\mathcal{S}}(x_{\mathcal{S}} ) = F_{\mathcal{S}}(x'_{\mathcal{S}} ) \}  \nonumber \\ \nonumber
& \phantom{-----------}  \times \Tr [  \sigma_E^{-1/4}  {\rho}^{x_{\mathcal{L}}}_{E} \sigma_E^{-1/2}  {\rho}^{x'_{\mathcal{L}}}_{E} \sigma_E^{-1/4} ] \nonumber \\ \nonumber 
&\\
& \leq \sum_{\mathcal{S} \subseteq \mathcal{L}}  \sum_{x_{\mathcal{L}}\in \mathcal{X}_{\mathcal{L}}} \!\!\! \sum_{\substack{x'_{\mathcal{L}}\in \mathcal{X}_{\mathcal{L}}\\ \textup{s.t } x'_{\mathcal{S}} \neq   x_{\mathcal{S}} \\ \phantom{\textup{s.t }} x'_{\mathcal{S}^c} =   x_{\mathcal{S}^c}}} \!\!\!\!\!\! 2^{-r_{\mathcal{S}}}   \Tr [  \sigma_E^{-1/4}  {\rho}^{x_{\mathcal{L}}}_{E} \sigma_E^{-1/2}  {\rho}^{x'_{\mathcal{L}}}_{E} \sigma_E^{-1/4}   ], \label{eq3}
\end{align}
where the inequality holds because $\mathbb{E}_{F_{\mathcal{L}}}  \mathds{1}\{ F_{\mathcal{S}}(x_{\mathcal{S}} ) = F_{\mathcal{S}}(x'_{\mathcal{S}} ) \} = \mathbb{E}_{F_{\mathcal{S}}}  \mathds{1}\{ F_{\mathcal{S}}(x_{\mathcal{S}} ) = F_{\mathcal{S}}(x'_{\mathcal{S}} ) \} = \prod_{l\in \mathcal{S}} \mathbb{E}_{F_{{l}}}  \mathds{1}\{ F_{{l}}(x_{{l}} ) = F_{{l}}(x'_{{l}} ) \} \leq \prod_{l\in \mathcal{S}} 2^{-r_l}$ by two-universality of the hash functions $F_{\mathcal{S}}$.
Note that we also~have
\begin{align} 
& \Tr [ \left( \sigma_E^{-1/4} {\rho}_{E} \sigma_E^{-1/4}  \right)^2 ] \nonumber \\
& = \sum_{\mathcal{S} \subseteq \mathcal{L}}  \sum_{x_{\mathcal{L}}\in \mathcal{X}_{\mathcal{L}}} \!\!\! \sum_{\substack{x'_{\mathcal{L}}\in \mathcal{X}_{\mathcal{L}}\\ \textup{s.t } x'_{\mathcal{S}} \neq   x_{\mathcal{S}} \\ \phantom{\textup{s.t }} x'_{\mathcal{S}^c} =   x_{\mathcal{S}^c}}} \Tr [  \sigma_E^{-1/4} {\rho}_{E}^{x_{\mathcal{L}}} \sigma_E^{-1/2}   {\rho}_{E}^{x'_{\mathcal{L}}} \sigma_E^{-1/4} ]. \label{eq4}
\end{align}
Hence, by combining \eqref{eq1}, \eqref{eq3}, and \eqref{eq4}, we have 
\begin{align}
&\lVert \rho_{A_{\mathcal{L}} E F_{\mathcal{L}}}  -  \rho_U \otimes \rho_{EF_{\mathcal{L}}} \rVert_1  \nonumber \\ \nonumber
& \leq  \sqrt{2^{r_{\mathcal{L}}} } \Big(  \sum_{\mathcal{S} \subseteq \mathcal{L}}  \sum_{x_{\mathcal{L}}\in \mathcal{X}_{\mathcal{L}}} \!\!\! \sum_{\substack{x'_{\mathcal{L}}\in \mathcal{X}_{\mathcal{L}}\\ \textup{s.t } x'_{\mathcal{S}} \neq   x_{\mathcal{S}} \\ \phantom{\textup{s.t }} x'_{\mathcal{S}^c} =   x_{\mathcal{S}^c}}} \!\!\!\!\!\! (2^{-r_{\mathcal{S}}} - 2^{-r_{\mathcal{L}}})   \nonumber \\ \nonumber
& \left. \phantom{---------} \times   \Tr [  \sigma_E^{-1/4}  {\rho}^{x_{\mathcal{L}}}_{E} \sigma_E^{-1/2}  {\rho}^{x'_{\mathcal{L}}}_{E} \sigma_E^{-1/4}   ]  \right)^{1/2}\\ \nonumber
& =   \Big( \sum_{\mathcal{S} \subsetneq \mathcal{L}}  \sum_{x_{\mathcal{S}}\in \mathcal{X}_{\mathcal{S}}}  \sum_{x_{\mathcal{S}^c}\in \mathcal{X}_{\mathcal{S}^c}}\!\!\! \sum_{\substack{x'_{\mathcal{S}}\in \mathcal{X}_{\mathcal{S}}\\ \textup{s.t } x'_{\mathcal{S}} \neq   x_{\mathcal{S}} }} \!\!\!\!\!\! (2^{r_{\mathcal{S}^c}} - 1)  \nonumber \\ \nonumber
& \left. \phantom{-----} \times  \Tr [  \sigma_E^{-1/4}  {\rho}^{(x_{\mathcal{S}},x_{\mathcal{S}^c})}_{E} \sigma_E^{-1/2}  {\rho}^{(x'_{\mathcal{S}},x_{\mathcal{S}^c})}_{E} \sigma_E^{-1/4}   ] \right)^{1/2}\\ \nonumber
& \stackrel{(a)}{\leq}  \Big(  \sum_{\mathcal{S} \subsetneq \mathcal{L}}  \sum_{x_{\mathcal{S}}\in \mathcal{X}_{\mathcal{S}}}  \sum_{x_{\mathcal{S}^c}\in \mathcal{X}_{\mathcal{S}^c}} \sum_{x'_{\mathcal{S}}\in \mathcal{X}_{\mathcal{S}}}  2^{r_{\mathcal{S}^c}}   \nonumber \\ \nonumber
& \left. \phantom{-----} \times  \Tr [  \sigma_E^{-1/4}  {\rho}^{(x_{\mathcal{S}},x_{\mathcal{S}^c})}_{E} \sigma_E^{-1/2}  {\rho}^{(x'_{\mathcal{S}},x_{\mathcal{S}^c})}_{E} \sigma_E^{-1/4}  ] \right)^{1/2}\\ \nonumber
& \stackrel{(b)}{=}   \sqrt{  \sum_{\mathcal{S} \subsetneq \mathcal{L}}  \sum_{x_{\mathcal{S}^c}\in \mathcal{X}_{\mathcal{S}^c}}  2^{r_{\mathcal{S}^c}}    \Tr [ \sigma_E^{-1/4}  {\rho}^{x_{\mathcal{S}^c}}_{E} \sigma_E^{-1/2}  {\rho}^{x_{\mathcal{S}^c}}_{E} \sigma_E^{-1/4}   ] }\\ \nonumber
& =   \sqrt{  \sum_{\mathcal{S} \subsetneq \mathcal{L}}   2^{r_{\mathcal{S}^c}}     \Tr [ \sum_{x_{\mathcal{S}^c}\in \mathcal{X}_{\mathcal{S}^c}} \ket{x_{\mathcal{S}^c}} \bra{x_{\mathcal{S}^c}} \otimes \left(  {\rho}^{x_{\mathcal{S}^c}}_{E} \sigma_E^{-1/2} \right)^2   ] }\\ \nonumber
& \stackrel{(c)}{=}   \sqrt{  \sum_{\mathcal{S} \subsetneq \mathcal{L}}    2^{r_{\mathcal{S}^c}} \Tr[\rho_{X_{\mathcal{S}^c}E}] 2^{-H_2(\rho_{X_{\mathcal{S}^c}E}|\sigma_E)}}\\ \nonumber
& \stackrel{(d)}{\leq}  \sqrt{  \sum_{\mathcal{S} \subsetneq \mathcal{L}}    2^{r_{\mathcal{S}^c}} 2^{- H_{\min}(\rho_{X_{\mathcal{S}^c}E}|\sigma_E)}}\\ \nonumber
& =  \sqrt{  \sum_{\substack{\mathcal{S} \subseteq \mathcal{L} \\ \mathcal{S} \neq \emptyset}}    2^{r_{\mathcal{S}}- H_{\min}(\rho_{X_{\mathcal{S}}E}|\sigma_E)}} , \nonumber
\end{align}
where $(a)$ holds because for any $x_{\mathcal{L}}\in \mathcal{X}_{\mathcal{L}}$, $x'_{\mathcal{S}}\in \mathcal{X}_{\mathcal{S}}$, $ \Tr [  \sigma_E^{-1/4}  {\rho}^{(x_{\mathcal{S}},x_{\mathcal{S}^c})}_{E} \sigma_E^{-1/2}  {\rho}^{(x'_{\mathcal{S}},x_{\mathcal{S}^c})}_{E} \sigma_E^{-1/4}   ] = \Tr [ \left(  \sigma_E^{-1/4}  {\rho}^{(x_{\mathcal{S}},x_{\mathcal{S}^c})}_{E} \sigma_E^{-1/4} \right) \left( \sigma_E^{-1/4}   {\rho}^{(x'_{\mathcal{S}},x_{\mathcal{S}^c})}_{E} \sigma_E^{-1/4}   \right) ] \geq 0$  since the trace of the product of two non-negative operators defined on the same Hilbert space is non-negative, $(b)$ holds with $ \forall \mathcal{S} \subseteq \mathcal{L}$, $\forall x_{\mathcal{S}} \in \mathcal{X}_{\mathcal{S}}$, $\rho_{E}^{x_{\mathcal{S}}} \triangleq \sum_{x_{\mathcal{S}^c}\in \mathcal{X}_{\mathcal{S}^c}} \rho_{E}^{x_{\mathcal{L}}} = p_{X_{\mathcal{S}}}(x_{\mathcal{S}}) \sum_{x_{\mathcal{S}^c}\in \mathcal{X}_{\mathcal{S}^c}} p_{X_{\mathcal{S}^c}|X_{\mathcal{S}}}(x_{\mathcal{S}^c}|x_{\mathcal{S}}) \bar{\rho}_{E}^{x_{\mathcal{L}}}$, $(c)$ holds by definition of the collision entropy in \eqref{defcol}, $(d)$ holds by Lemma \ref{lem_app3} in Appendix~\ref{App_6}.

\section{Proof of Lemma \ref{lemprod}} \label{app_2}
Define
\begin{align*}
\mathcal{A}&\triangleq \left\{ (x^n_{\mathcal{L}},y^n) \in \mathcal{X}^n_{\mathcal{L}} \times \mathcal{Y}^n : \forall \mathcal{S} \subseteq \mathcal{L}, \right.\\
& \left. \phantom{---} -\log p_{X^n_{\mathcal{S}}Y^n} (x^n_{\mathcal{S}},y^n)  \geq H(X^n_{\mathcal{S}}Y^n) - n \delta_{\mathcal{S}}(n) \right\},\\
\mathcal{B}&\triangleq \left\{ y^n \in \mathcal{Y}^n : -\log p_{Y^n} (y^n)  \leq H(Y^n) + n \delta(n) \right\},
\end{align*}
and for $\mathcal{S} \subseteq \mathcal{L}$,
\begin{align*} 
&\mathcal{A}_{\mathcal{S}} \triangleq \left\{ (x^n_{\mathcal{S}}, y^n) \in \mathcal{X}^n_{\mathcal{S}} \times \mathcal{Y}^n :\right.\\
& \left. \phantom{---} -\log p_{X^n_{\mathcal{S}}Y^n} (x^n_{\mathcal{S}},y^n)  \geq H(X^n_{\mathcal{S}}Y^n) - n \delta_{\mathcal{S}}(n)  \right\}.
\end{align*}
Next, define for $(x^n_{\mathcal{L}},y^n) \in \mathcal{X}^n_{\mathcal{L}} \times \mathcal{Y}^n$,
\begin{align} 
& q_{X^n_{\mathcal{L}}Y^n}(x^n_{\mathcal{L}},y^n) \nonumber \\
& \phantom{--} \triangleq \mathds{1} \{ (x^n_{\mathcal{L}},y^n) \in \mathcal{A}  \} \mathds{1} \{ y^n \in \mathcal{B}  \} p_{X^n_{\mathcal{L}}Y^n}(x^n_{\mathcal{L}},y^n), \label{equnn}
\end{align}
and for $\mathcal{S} \subseteq \mathcal{L}$,
\begin{align} \label{eqmarg}
q_{X^n_{\mathcal{S}}Y^n}(x^n_{\mathcal{S}},y^n) \triangleq \sum_{x^n_{\mathcal{S}^c} \in \mathcal{X}^n_{\mathcal{S}^c} } q_{X^n_{\mathcal{L}}Y^n}(x^n_{\mathcal{L}},y^n).
\end{align}
We first show that $\mathbb{V}(p_{X^n_{\mathcal{L}}Y^n},q_{X^n_{\mathcal{L}}Y^n}) \leq \epsilon$. We have
\begin{align*}
& \mathbb{V}(p_{X^n_{\mathcal{L}}Y^n},q_{X^n_{\mathcal{L}}Y^n}) \\
& = \sum_{x^n_{\mathcal{L}},y^n} |p_{X^n_{\mathcal{L}}Y^n}(x^n_{\mathcal{L}},y^n)- q_{X^n_{\mathcal{L}}Y^n}(x^n_{\mathcal{L}},y^n)| \\
& \leq \sum_{x^n_{\mathcal{L}},y^n} p_{X^n_{\mathcal{L}}Y^n}(x^n_{\mathcal{L}},y^n) (\mathds{1} \{ (x^n_{\mathcal{L}},y^n) \notin  \mathcal{A} \}  + \mathds{1} \{ y^n \notin  \mathcal{B} \}) \\
& = \mathbb{P} \left[ (X^n_{\mathcal{L}},Y^n) \notin  \mathcal{A} \right]  + \mathbb{P} \left[ Y^n \notin  \mathcal{B} \right]\\
& = \mathbb{P} \left[ \exists \mathcal{S} \subseteq \mathcal{L}, (X^n_{\mathcal{S}},Y^n) \notin  \mathcal{A}_{\mathcal{S}} \right] + \mathbb{P} \left[ Y^n \notin  \mathcal{B} \right] \\
& \stackrel{(a)}{\leq} \sum_{ \mathcal{S} \subseteq \mathcal{L} }\mathbb{P} \left[  (X^n_{\mathcal{S}},Y^n) \notin  \mathcal{A}_{\mathcal{S}} \right] +  \mathbb{P} \left[ Y^n \notin  \mathcal{B} \right]\\
& \stackrel{(b)}{\leq} \sum_{ \mathcal{S} \subseteq \mathcal{L} } 2^{-\frac{n \delta_{\mathcal{S}}^2(n)}{2\log(|\mathcal{X}_{\mathcal{S}}||\mathcal{Y}| +3)^2}} + 2^{-\frac{n \delta^2 (n)}{2\log(|\mathcal{Y}|+3)^2}}\\
& \stackrel{(c)}= \sum_{ \mathcal{S} \subseteq \mathcal{L}} 2^{-L}\epsilon/2 + \epsilon/2\\
& = \epsilon,
\end{align*}
where $(a)$ holds by the union bound, $(b)$ holds by Lemma \ref{lem_apphr} in Appendix \ref{App_6}, $(c)$ holds by definitions of $\delta_{\mathcal{S}}(n)$ and $\delta(n)$. Next, for $\mathcal{S} \subseteq \mathcal{L}$, we have
\begin{align*}
& H_{\min}(q_{X^n_{\mathcal{S}}Y^n}) \\
&= -  \max_{ (x^n_{\mathcal{S}},y^n) \in \mathcal{X}^n_{\mathcal{S}} \times \mathcal{Y}^n} \log {q_{X^n_{\mathcal{S}}Y^n}(x^n_{\mathcal{S}},y^n)} \\
& \stackrel{(a)}{=} -  \max_{(x^n_{\mathcal{S}},y^n) \in \mathcal{X}^n_{\mathcal{S}} \times \mathcal{Y}^n} \log \left( \textstyle\sum_{x^n_{\mathcal{S}^c} \in \mathcal{X}^n_{\mathcal{S}^c} } \mathds{1} \{ (x^n_{\mathcal{L}},y^n) \in \mathcal{A}  \} \right. \\
& \left. \phantom{------------} \times \mathds{1} \{ y^n \in \mathcal{B}  \} p_{X^n_{\mathcal{L}}Y^n}(x^n_{\mathcal{L}},y^n) \right)\\
& \stackrel{(b)}{\geq} -  \max_{(x^n_{\mathcal{S}},y^n) \in \mathcal{X}^n_{\mathcal{S}} \times \mathcal{Y}^n} \log ({  \mathds{1} \{ (x^n_{\mathcal{S}},y^n) \in \mathcal{A}_{\mathcal{S}}  \} p_{X^n_{\mathcal{S}}Y^n}(x^n_{\mathcal{S}},y^n)})\\
& \stackrel{(c)}{\geq} H(X^n_{\mathcal{S}}Y^n) - n \delta_{\mathcal{S}}(n) ,
\end{align*}
where $(a)$ holds by \eqref{equnn} and \eqref{eqmarg}, $(b)$ holds because for any $(x^n_{\mathcal{L}},y^n) \in \mathcal{X}^n_{\mathcal{L}} \times \mathcal{Y}^n$, $ \mathds{1} \{ (x^n_{\mathcal{S}}, y^n) \in \mathcal{A}_{\mathcal{S}}  \} \geq \mathds{1} \{ (x^n_{\mathcal{L}},y^n) \in \mathcal{A}  \} \mathds{1} \{ y^n \in \mathcal{B}  \}$ and by marginalization over $X^n_{\mathcal{S}^c}$, $(c)$ holds by definition of $\mathcal{A}_{\mathcal{S}}$. Then, we also have
\begin{align*}
 H_{\max}(q_{Y^n}) 
& = \log \textup{supp}(q_{Y^n})\\
&  \stackrel{(a)}\leq \log |\mathcal{B}|\\
& \stackrel{(b)}\leq n H(Y) + n \delta(n),
\end{align*}
where $(a)$ holds by \eqref{equnn} and \eqref{eqmarg}, and $(b)$ holds because $1 \geq \sum_{y^n \in \mathcal{B}} p_{Y^n}(y^n) \geq |\mathcal{B}| 2^{- H(Y^n) - n \delta(n)}$ by definition of~$\mathcal{B}$.

\section{Proof of Lemma \ref{lemlhlprod}} \label{app_3}
Consider a spectral decomposition of the product state $\rho_{X^n_{\mathcal{L}}E^n}  $ given by $$\rho_{X^n_{\mathcal{L}}E^n}   = \sum_{x^n_{\mathcal{L}},e^n} p_{X^n_{\mathcal{L}}E^n}(x^n_{\mathcal{L}},e^n)\ket{\phi_{x^n_{\mathcal{L}},e^n}}\!\bra{\phi_{x^n_{\mathcal{L}},e^n}}.$$ By Lemma \ref{lemprod}, there exists a subnormalized non-negative function $q_{X^n_{\mathcal{L}}E^n}$ such that $\mathbb{V}(p_{X^n_{\mathcal{L}}E^n},q_{X^n_{\mathcal{L}}E^n}) \leq \epsilon$ and
\begin{align}
\forall \mathcal{S} \subseteq \mathcal{L}, H_{\min}(q_{X^n_{\mathcal{S}}E^n}) & \geq n H(X_{\mathcal{S}}E) - n \delta_{\mathcal{S}}(n), \label{eq1a}\\
H_{\max}(q_{E^n}) & \leq n H(E) + n \delta(n).  \label{eq1b}
\end{align}
Next, define the state
\begin{align}
\bar{\rho}_{X^n_{\mathcal{L}}E^n}   = \sum_{x^n_{\mathcal{L}},e^n} q_{X^n_{\mathcal{L}}E^n}(x^n_{\mathcal{L}},e^n)\ket{\phi_{x^n_{\mathcal{L}},e^n}}\!\bra{\phi_{x^n_{\mathcal{L}},e^n}}, \label{eqrhobar}
\end{align}
and for $\mathcal{S} \subseteq \mathcal{L}$
\begin{align}
\bar{\rho}_{X^n_{\mathcal{S}}E^n}  
& = \Tr_{X^n_{\mathcal{S}^c}} [\bar{\rho}_{X^n_{\mathcal{L}}E^n}] \nonumber \\
& =  \sum_{x^n_{\mathcal{S}},e^n} q_{X^n_{\mathcal{S}}E^n}(x^n_{\mathcal{S}},e^n)\ket{\phi_{x^n_{\mathcal{S}},e^n}}\!\bra{\phi_{x^n_{\mathcal{S}},e^n}},\label{eqrhobar2}
\end{align}
where for any $(x^n_{\mathcal{S}},e^n)$, $q_{X^n_{\mathcal{S}}E^n}(x^n_{\mathcal{S}},e^n) \triangleq \sum_{x^n_{\mathcal{S}^c}}q_{X^n_{\mathcal{L}}E^n}(x^n_{\mathcal{L}},e^n)$.
Hence, we have
\begin{align}
&\lVert \rho_{X^n_{\mathcal{L}}E^n }  -  
\bar{\rho}_{X^n_{\mathcal{L}}E^n } \rVert_1  \nonumber\\
& \leq \sum_{x^n_{\mathcal{L}},e^n} |q_{X^n_{\mathcal{L}}E^n}(x^n_{\mathcal{L}},e^n) - p_{X^n_{\mathcal{L}}E^n}(x^n_{\mathcal{L}},e^n)|  \nonumber\\ \nonumber
& =  \mathbb{V}(p_{X^n_{\mathcal{L}}E^n}, q_{X^n_{\mathcal{L}}E^n}) \\
& \leq \epsilon. \label{eqepsilon}
\end{align}
Then, let $\rho_U$ be the fully mixed state on $\mathcal{H}_{F_{\mathcal{L}}(X^n_{\mathcal{L}})}$, and define the operator $\bar{\rho}_{F_{\mathcal{L}}(X^n_{\mathcal{L}})E^n F_{\mathcal{L}}}$ as in~\eqref{rhoxelhl} using $\bar{\rho}_{X^n_{\mathcal{L}}E^n}$ in place of ${\rho}_{X^n_{\mathcal{L}}E^n}$. We have 
\begin{align*}
&\lVert \rho_{F_{\mathcal{L}}(X^n_{\mathcal{L}})E^n F_{\mathcal{L}}}  -  \rho_U \otimes \rho_{E^nF_{\mathcal{L}}} \rVert_1 \\
&\stackrel{(a)}{\leq} \lVert \rho_{F_{\mathcal{L}}(X^n_{\mathcal{L}})E^n F_{\mathcal{L}}}  -  \bar{\rho}_{F_{\mathcal{L}}(X^n_{\mathcal{L}})E^n F_{\mathcal{L}}} \rVert_1\\
& \phantom{--} + \lVert \bar{\rho}_{F_{\mathcal{L}}(X^n_{\mathcal{L}})E^n F_{\mathcal{L}}}  -  {\rho}_U \otimes \bar{\rho}_{E^nF_{\mathcal{L}}} \rVert_1 \\
& \phantom{--}+ \lVert  {\rho}_U \otimes \bar{\rho}_{E^nF_{\mathcal{L}}} - \rho_U \otimes \rho_{E^nF_{\mathcal{L}}} \rVert_1 \\
&\stackrel{(b)}{\leq} 2 \epsilon + \lVert \bar{\rho}_{F_{\mathcal{L}}(X^n_{\mathcal{L}})E^n F_{\mathcal{L}}}  -  {\rho}_U \otimes \bar{\rho}_{E^nF_{\mathcal{L}}} \rVert_1 \\
& \stackrel{(c)}{\leq} 2 \epsilon  + \sqrt{  \sum_{\mathcal{S} \subseteq \mathcal{L}, \mathcal{S} \neq \emptyset}    2^{r_{\mathcal{S}}- H_{\min}(\bar{\rho}_{X^n_{\mathcal{S}}E^n}|\sigma_{E^n})}} \\
& \stackrel{(d)}{=} 2 \epsilon  + \sqrt{  \sum_{\mathcal{S} \subseteq \mathcal{L}, \mathcal{S} \neq \emptyset}    2^{r_{\mathcal{S}}- H_{\min}(\bar{\rho}_{X^n_{\mathcal{S}}E^n}) + H_{\max}(\bar{\rho}_{E^n})}} \\
& \stackrel{(e)} = 2 \epsilon  + \sqrt{  \sum_{\mathcal{S} \subseteq \mathcal{L}, \mathcal{S} \neq \emptyset}    2^{r_{\mathcal{S}} + \log (\lambda_{\max}(\bar{\rho}_{X^n_{\mathcal{S}}E^n})) + \log (\rank(\bar{\rho}_{E^n}))}}\\
&\stackrel{(f)} = 2 \epsilon  + \sqrt{  \sum_{\mathcal{S} \subseteq \mathcal{L}, \mathcal{S} \neq \emptyset}    2^{r_{\mathcal{S}} - H_{\min}(q_{X^n_{\mathcal{S}}E^n}) + H_{\max} (q_{E^n})}} \\
& \stackrel{(g)}{\leq} 2 \epsilon  + \sqrt{  \sum_{\mathcal{S} \subseteq \mathcal{L}, \mathcal{S} \neq \emptyset}    2^{r_{\mathcal{S}} - n ( H(X_{\mathcal{S}}E) - H(E) - \delta_{\mathcal{S}}(n) - \delta(n))} }\\
& \stackrel{(i)}= 2 \epsilon  + \sqrt{  \sum_{\mathcal{S} \subseteq \mathcal{L}, \mathcal{S} \neq \emptyset}    2^{r_{\mathcal{S}} - n H({X_{\mathcal{S}}|E})_{\rho} +n( \delta_{\mathcal{S}}(n) + \delta(n))} },
\end{align*}
where $(a)$ holds by the triangle inequality, $(b)$ holds by the data processing inequality, e.g., \cite[Lemma~A.2.1]{renner2008security},  and~\eqref{eqepsilon}, $(c)$ holds by Lemma \ref{lemlhl} where $\sigma_{E^n}$ is the fully mixed state on the support of $\bar{\rho}_{E^n}$, $(d)$ holds by Lemma \ref{lem_app4} in Appendix~\ref{App_6}, $(e)$ follows from  the definitions of $H_{\min}$ and $H_{\max}$, where $\lambda_{\max}(\bar{\rho}_{X^n_{\mathcal{S}}E^n})$ is the maximum eigenvalue of $\bar{\rho}_{X^n_{\mathcal{S}}E^n}$, $(f)$~holds by \eqref{eqrhobar} and \eqref{eqrhobar2}, $(g)$ holds by \eqref{eq1a},~\eqref{eq1b}, $(i)$~holds because the von Neumann entropy of an operator with eigenvalues $(p_i)$ is equal to the Shannon entropy of a random variable distributed according to~$(p_i)$.

\section{Proof of Lemma \ref{lemMac}} \label{App_5}
Assume that  a $(2^{nR^{\textup{DC}}_l})_{l\in\mathcal{L}}$ distributed source code is given and that the corresponding encoding and decoding functions are $(g_l)_{l\in\mathcal{L}}$ and $h$, respectively. We use the same notation as in Definition \ref{defsw}. To simplify notation, we define $\mathbf{u}_{\mathcal{L}} \triangleq {u}_{\mathcal{L}}^n$, for ${u}_{\mathcal{L}}^n \in \mathcal{U}_{\mathcal{L}}^n$. By definition, we have $\lim_{n\to \infty}P_e(n) = 0$ and
\begin{align*}
& P_e(n)\\
&  \triangleq \frac{1}{|\mathcal{U}_{\mathcal{L}}^n|} \sum_{\mathbf{u}_{\mathcal{L}}\in \mathcal{U}^n_{\mathcal{L}}} \mathbb{P}\left[\mathbf{u}_{\mathcal{L}} \neq h\left( \bar{\rho}_{B^n}^{\mathbf{u}_{\mathcal{L}}}, g_{\mathcal{L}}(\mathbf{u}_{\mathcal{L}}) \right)\right] \\
&  = \frac{1}{|\mathcal{U}_{\mathcal{L}}^n|} \sum_{c_{\mathcal{L}} \in \mathcal{C}_{\mathcal{L}}}  \sum_{ \mathbf{u}_{\mathcal{L}} \in g_{\mathcal{L}}^{-1}(c_{\mathcal{L} })} \mathbb{P}\left[\mathbf{u}_{\mathcal{L}} \neq h\left( \bar{\rho}_{B^n}^{\mathbf{u}_{\mathcal{L}}}, g_{\mathcal{L}}(\mathbf{u}_{\mathcal{L}}) \right)\right]\\
&  = \frac{1}{|\mathcal{C}_{\mathcal{L}}|} \sum_{c_{\mathcal{L}} \in \mathcal{C}_{\mathcal{L}}}   \sum_{ \mathbf{u}_{\mathcal{L}} \in g_{\mathcal{L}}^{-1}(c_{\mathcal{L} })} \prod_{l \in \mathcal{L}} \frac{\mathbb{P}\left[\mathbf{u}_{\mathcal{L}} \neq h\left( \bar{\rho}_{B^n}^{\mathbf{u}_{\mathcal{L}}}, g_{\mathcal{L}}(\mathbf{u}_{\mathcal{L}}) \right)\right]}{|\mathcal{U}^n_{l}|/|\mathcal{C}_{l}|} \\
& \stackrel{(a)} \geq \frac{1}{|\mathcal{C}_{\mathcal{L}}|} \sum_{c_{\mathcal{L}} \in \mathcal{C}'_{\mathcal{L}}}   \sum_{ \mathbf{u}_{\mathcal{L}} \in g_{\mathcal{L}}^{-1}(c_{\mathcal{L} })} \prod_{l \in \mathcal{L}} \frac{\mathbb{P}\left[\mathbf{u}_{\mathcal{L}} \neq h\left( \bar{\rho}_{B^n}^{\mathbf{u}_{\mathcal{L}}}, g_{\mathcal{L}}(\mathbf{u}_{\mathcal{L}}) \right)\right]}{|\mathcal{U}^n_{l}|/|\mathcal{C}_{l}|} \\
& \stackrel{(b)} \geq \frac{1}{|\mathcal{C}_{\mathcal{L}}|} \sum_{c_{\mathcal{L}} \in \mathcal{C}'_{\mathcal{L}}}   \sum_{ \mathbf{u}_{\mathcal{L}} \in g_{\mathcal{L}}^{-1}(c_{\mathcal{L} })} \prod_{l \in \mathcal{L}} \frac{\epsilon \mathbb{P}\left[\mathbf{u}_{\mathcal{L}} \neq h\left( \bar{\rho}_{B^n}^{\mathbf{u}_{\mathcal{L}}}, g_{\mathcal{L}}(\mathbf{u}_{\mathcal{L}}) \right)\right]}{|g_l^{-1}(c_l)|}  \\ %
&  \stackrel{(c)} = \frac{\epsilon}{|\mathcal{C}_{\mathcal{L}}|} \sum_{c_{\mathcal{L}} \in \mathcal{C}'_{\mathcal{L}}}    \mathbb{E}_{p_{\mathbf{U}_{\mathcal{L}}|C_{\mathcal{L}}=c_{\mathcal{L}}}} \mathbb{P}\left[\mathbf{u}_{\mathcal{L}} \neq h\left( \bar{\rho}_{B^n}^{\mathbf{u}_{\mathcal{L}}}, g_{\mathcal{L}}(\mathbf{u}_{\mathcal{L}}) \right)\right]\\
& \stackrel{(d)} \geq \epsilon  \mathbb{E}_{p_{\mathbf{U}_{\mathcal{L}}|C_{\mathcal{L}}=c^*_{\mathcal{L}}}} \mathbb{P}\left[\mathbf{u}_{\mathcal{L}} \neq h\left( \bar{\rho}_{B^n}^{\mathbf{u}_{\mathcal{L}}}, g_{\mathcal{L}}(\mathbf{u}_{\mathcal{L}}) \right)\right] \sum_{c_{\mathcal{L}} \in \mathcal{C}'_{\mathcal{L}}}  \frac{1}{|\mathcal{C}_{\mathcal{L}}|}\\
& = \epsilon  \mathbb{E}_{p_{\mathbf{U}_{\mathcal{L}}|C_{\mathcal{L}}=c^*_{\mathcal{L}}}} \mathbb{P}\left[\mathbf{u}_{\mathcal{L}} \neq h\left( \bar{\rho}_{B^n}^{\mathbf{u}_{\mathcal{L}}}, g_{\mathcal{L}}(\mathbf{u}_{\mathcal{L}}) \right)\right] \\
& \phantom{-----} \times \sum_{c_{\mathcal{L}} \in \mathcal{C}_{\mathcal{L}}}  \frac{1}{|\mathcal{C}_{\mathcal{L}}|} \mathds{1} \{|g_l^{-1}(c_l)|\geq \epsilon |\mathcal{U}^n_{l}|/|\mathcal{C}_{l}| , \forall l \in \mathcal{L} \} \\
& \stackrel{(e)} \geq \epsilon (1-\epsilon)^L \mathbb{E}_{p_{\mathbf{U}_{\mathcal{L}}|C_{\mathcal{L}}=c^*_{\mathcal{L}}}} \mathbb{P}\left[\mathbf{u}_{\mathcal{L}} \neq h\left( \bar{\rho}_{B^n}^{\mathbf{u}_{\mathcal{L}}}, g_{\mathcal{L}}(\mathbf{u}_{\mathcal{L}}) \right)\right], \numberthis \label{eqPe}
\end{align*}
where in $(a)$ we have defined $\mathcal{C}'_{\mathcal{L}} \triangleq \{ c_{\mathcal{L}} \in \mathcal{C}_{\mathcal{L}}:|g_l^{-1}(c_l)|\geq \epsilon |\mathcal{U}^n_{l}|/|\mathcal{C}_{l}| , \forall l \in \mathcal{L}\}$, $(b)$ holds by definition of $\mathcal{C}'_{\mathcal{L}}$, in $(c)$ we have defined $p_{\mathbf{U}_{\mathcal{L}}|C_{\mathcal{L}}=c_{\mathcal{L}}} \triangleq \prod_{l\in\mathcal{L}}p_{\mathbf{U}_{l}|C_{l}=c_{l}} $ and $p_{\mathbf{U}_{l}|C_{l}=c_{l}}$ is the uniform distribution over $g^{-1}_l(c_l)$, $l\in\mathcal{L}$, in $(d)$ we have chosen $c^*_{\mathcal{L}} \in \argmin_{c_{\mathcal{L}}} \mathbb{E}_{p_{\mathbf{U}_{\mathcal{L}}|C_{\mathcal{L}}=c_{\mathcal{L}}}}\mathbb{P}\left[\mathbf{u}_{\mathcal{L}} \neq h\left( \bar{\rho}_{B^n}^{\mathbf{u}_{\mathcal{L}}}, g_{\mathcal{L}}(\mathbf{u}_{\mathcal{L}}) \right)\right]$, $(e)$ holds by   Lemma \ref{lempreimage} in Appendix \ref{App_6}.

From \eqref{eqPe}, we conclude that 
\begin{align}
\lim_{n\to \infty} \mathbb{E}_{p_{\mathbf{U}_{\mathcal{L}}|C_{\mathcal{L}}=c^*_{\mathcal{L}}}} \mathbb{P}\left[\mathbf{u}_{\mathcal{L}} \neq h\left( \bar{\rho}_{B^n}^{\mathbf{u}_{\mathcal{L}}}, g_{\mathcal{L}}(\mathbf{u}_{\mathcal{L}}) \right)\right] = 0. \label{eqPemac}
\end{align}
For $l\in\mathcal{L}$, let $\mathcal{M}_l$ be such that $|\mathcal{M}_l|=|g^{-1}_l(c_l^*)|$, and let the encoder $e_l$ be a bijection between  $\mathcal{M}_l$ and $g^{-1}_l(c_l^*)$.
Hence, for any $m_{\mathcal{L}} \in \mathcal{M}_{\mathcal{L}}$, we have $g_{\mathcal{L}}(e_{\mathcal{L}}(m_{\mathcal{L}})) = c^*_{\mathcal{L}}$. Then, define the decoder as $d(  \bar{\rho}_{B^n}^{e_{\mathcal{L}}(M_{\mathcal{L}})}) \triangleq e^{-1}_{\mathcal{L}}\left(h(\bar{\rho}_{B^n}^{e_{\mathcal{L}}(M_{\mathcal{L}})},c^*_{\mathcal{L}}) \right)$.
Hence, by \eqref{eqPemac}, we have $\lim_{n\to \infty}\mathbb{P}[d( \bar{\rho}_{B^n}^{e_{\mathcal{L}}(M_{\mathcal{L}})} ) \neq M_{\mathcal{L}}] = 0$. Finally, for $l\in\mathcal{L}$, we have $2^{nR_l}=|\mathcal{M}_l| \geq \epsilon |\mathcal{U}_{l}^n|/|\mathcal{C}_{l}|  = \epsilon 2^{n(R_l^{\textup{U}}- R_l^{\textup{DC}})}$, which yields $ R_l \geq R^{\textup{U}}_l - R^{\textup{DC}}_l$ as $n\to \infty$. %

\section{Proof of Lemma \ref{lemsigmaP2}} \label{App_8}
The arguments closely follow the proof for the special case $L=1$, e.g., \cite[Th. 13.6.2]{wilde2013quantum}. We first prove $P_{\textup{MAC}}^{\textup{sum}}  (\mathcal{N}) 
 \leq Q^{\textup{sum}}_{\textup{MAC}} (\mathcal{N} )$. 	Consider a state $\rho_{X_{\mathcal{L}} EB}\triangleq  \mathcal{U}^{\mathcal{N}}_{A'_{\mathcal{L}} \to BE} (\rho_{X_{\mathcal{L}} A'_{\mathcal{L}}} )$ that  achieves $P_{\textup{MAC}}^{\textup{sum}}  (\mathcal{N})$, i.e., maximizes the right-hand side in~\eqref{eqPmacsum}.
	For $l \in \mathcal{L}$, consider a spectral decomposition for $\rho_{A_l'}^{x_l} = \sum_{y_l} p(y_l|x_l) \psi_{A_l'}^{x_l,y_l}$, where each state $\psi_{A_l'}^{x_l,y}$ is pure. Next, consider $\sigma_{X_{\mathcal{L}}Y_{\mathcal{L}} B E}$ such that $\Tr_{Y_{\mathcal{L}}} [\sigma_{X_{\mathcal{L}}Y_{\mathcal{L}} B E}] =\rho_{X_{\mathcal{L}}B E}$~with 
	\begin{align*}
&	\sigma_{X_{\mathcal{L}}Y_{\mathcal{L}} B E} \triangleq \sum_{x_{\mathcal{L}}}\sum_{y_{\mathcal{L}}} p_{X_{\mathcal{L}}}(x_{\mathcal{L}}) p_{Y_{\mathcal{L}}|X_{\mathcal{L}}} (y_{\mathcal{L}}|x_{\mathcal{L}}) \\
		& \phantom{------}\ket{x_{\mathcal{L}}} \bra{x_{\mathcal{L}}} \otimes \ket{y_{\mathcal{L}}} \bra{y_{\mathcal{L}}} \otimes \mathcal{U}^{\mathcal{N}}_{A'_{\mathcal{L}} \to B E} ( \bigotimes_{l\in \mathcal{L}} \psi_{A_l'}^{x_l,y_l}),
	\end{align*}
	where $p_{X_{\mathcal{L}}} (x_{\mathcal{L}}) \triangleq \prod_{l \in \mathcal{L}} p(x_l)$, $p_{Y_{\mathcal{L}}|X_{\mathcal{L}}} (y_{\mathcal{L}}|x_{\mathcal{L}}) \triangleq \prod_{l \in \mathcal{L}} p(y_l|x_l)$, $x_{\mathcal{L}} \triangleq (x_l)_{l\in \mathcal{L}}$, $y_{\mathcal{L}} \triangleq (y_l)_{l\in \mathcal{L}}$, $\ket{x_{\mathcal{L}}} \bra{x_{\mathcal{L}}}  \triangleq \bigotimes_{l\in\mathcal{L}}\ket{x_{l}} \bra{x_{l}} $, and $\ket{y_{\mathcal{L}}} \bra{y_{\mathcal{L}}}  \triangleq \bigotimes_{l\in\mathcal{L}}\ket{y_{l}} \bra{y_{l}} $. Then, we have
	\begin{align*}
		&P_{\textup{MAC}}^{\textup{sum}}  (\mathcal{N}) \\
		& =  I(X_{\mathcal{L}};B)_{\rho} - I(X_{\mathcal{L}};E)_{\rho} \\
		& =  I(X_{\mathcal{L}};B)_{\sigma} - I(X_{\mathcal{L}};E)_{\sigma} \\
		& =  I(X_{\mathcal{L}}Y_{\mathcal{L}};B)_{\sigma} - I(X_{\mathcal{L}}Y_{\mathcal{L}};E)_{\sigma} \\
		& \phantom{--} - I(Y_{\mathcal{L}};B|X_{\mathcal{L}})_{\sigma} + I(Y_{\mathcal{L}};E|X_{\mathcal{L}})_{\sigma} \\
		& \stackrel{(a)} \leq I(X_{\mathcal{L}}Y_{\mathcal{L}};B)_{\sigma} - I(X_{\mathcal{L}}Y_{\mathcal{L}};E)_{\sigma} \\
		& = H(B)_{\sigma} - H(E)_{\sigma} + H(E|X_{\mathcal{L}}Y_{\mathcal{L}})_{\sigma} - H(B|X_{\mathcal{L}}Y_{\mathcal{L}})_{\sigma}\\
				&\stackrel{(b)} = H(B)_{\sigma} - H(E)_{\sigma}\\				
		&\stackrel{(c)} = H(B)_{\phi} - H(A_{\mathcal{L}}B)_{\phi}\\
	&  = I(A_{\mathcal{L}} \rangle B)_{\phi}\\
		&\stackrel{(d)} \leq Q^{\textup{sum}}_{\textup{MAC}} (\mathcal{N} ),
\end{align*}
where $(a)$ holds by the quantum data processing inequality because $\mathcal{N}$ is degradable, $(b)$ holds because $\sigma_{BE}^{x_{\mathcal{L}},y_{\mathcal{L}}}$ is pure by purity of $ \bigotimes_{l\in \mathcal{L}} \psi_{A_l'}^{x_l,y_l}$, in $(c)$, for $l\in\mathcal{L}$, we consider $\phi_{A_lA'_l}$ a purification of $\rho_{A'_l}$ and define $\phi_{A_{\mathcal{L}}A'_{\mathcal{L}}} \triangleq \bigotimes_{l \in \mathcal{L}} \phi_{A_lA'_l}$ and $\phi_{A_{\mathcal{L}}BE} \triangleq \mathcal{U}^{\mathcal{N}}_{A'_{\mathcal{L}} \to BE} (\phi_{A_{\mathcal{L}}A'_{\mathcal{L}}} )$ such that $\phi_{A_{\mathcal{L}}BE}$ is pure and $\Tr_{A_{\mathcal{L}}} [\phi_{A_{\mathcal{L}}BE}]= \rho_{BE} = \sigma_{BE}$, $(d)$ holds by definition of~$Q^{\textup{sum}}_{\textup{MAC}} (\mathcal{N} )$.

Next, we show $P_{\textup{MAC}}^{\textup{sum}}  (\mathcal{N}) 
 \geq Q^{\textup{sum}}_{\textup{MAC}} (\mathcal{N} )$. Consider a state $\phi_{A_{\mathcal{L}}BE} \triangleq \mathcal{U}^{\mathcal{N}}_{A'_{\mathcal{L}} \to BE} (\phi_{A_{\mathcal{L}}A'_{\mathcal{L}}} )$ that achieves $Q^{\textup{sum}}_{\textup{MAC}} (\mathcal{N} )$, i.e., maximizes the right-hand side of \eqref{eqqsum}. Consider for $l\in\mathcal{L}$ a spectral decomposition of  $\phi_{A_{l}'}$ such that $\phi_{A_{l}'} = \sum_{x_l} p_{X_l}(x_l) \phi_{A_{l}'}^{x_l} $, where each state $\phi_{A_{l}'}^{x_l}$ is pure. Then, define $$\sigma_{X_{\mathcal{L}} A'_{\mathcal{L}}} \triangleq \sum_{x_{\mathcal{L}}} p_{X_{\mathcal{L}}}(x_{\mathcal{L}}) \ket{x_{\mathcal{L}}} \bra{x_{\mathcal{L}}} \otimes \bigotimes_{l \in \mathcal{L}} \phi_{A_{l}'}^{x_l} ,$$ where $\ket{x_{\mathcal{L}}} \bra{x_{\mathcal{L}}} \triangleq \bigotimes_{\l\in\mathcal{L}} \ket{x_{l}} \bra{x_{l}}$, $x_{\mathcal{L}} \triangleq (x_l)_{l \in \mathcal{L}}$, and $p_{X_{\mathcal{L}}}(x_{\mathcal{L}}) \triangleq \prod_{\l\in\mathcal{L}}p_{X_l}(x_l)$. Define also $\sigma_{X_{\mathcal{L}}BE} \triangleq  \mathcal{U}^{\mathcal{N}}_{A'_{\mathcal{L}} \to BE} (\sigma_{X_{\mathcal{L}} A'_{\mathcal{L}}})$. Then, we have
 \begin{align*}
 Q^{\textup{sum}}_{\textup{MAC}} (\mathcal{N} ) 
 & = I(A_{\mathcal{L}} \rangle B)_{\phi} \\
 &\stackrel{(a)} = H(B)_{\phi} - H(E)_{\phi}\\
  & = H(B)_{\sigma} - H(E)_{\sigma} \\
  &\stackrel{(b)} =  I(X_{\mathcal{L}};B)_{\sigma} - I(X_{\mathcal{L}};E)_{\sigma} \\
    &\stackrel{(c)} \leq  P_{\textup{MAC}}^{\textup{sum}}  (\mathcal{N}),
 \end{align*}
 where $(a)$ holds because $H(A_{\mathcal{L}} B)_{\rho} = H(E)_{\rho}$ by purity of  $\phi_{A_{\mathcal{L}}BE}$, $(b)$ holds because $H(E|X_{\mathcal{L}})_{\sigma}= H(B|X_{\mathcal{L}})_{\sigma}$ by purity of $\sigma_{BE}^{x_{\mathcal{L}}}$, $(c)$ holds by definition of $P_{\textup{MAC}}^{\textup{sum}}  (\mathcal{N})$.

\section{Supporting lemmas} \label{App_6}

\begin{lem}[{\cite[Lemma 5.1.3]{renner2008security}}] \label{lem_app}
Let $\rho$ be a Hermitian operator and $\sigma$ be a nonnegative operator on the same Hilbert space. Then,
$
\lVert \rho \rVert_1 \leq \sqrt{ \Tr[\sigma] \Tr[ (\rho \sigma^{-1/2})^2]}.
$
\end{lem}

\begin{lem}[{\cite[Lemma B.5.3]{renner2008security}}] \label{lem_app3}
For any $\rho_{XE} \in \mathcal{S}_{\leq}( \mathcal{H}_X \otimes \mathcal{H}_E)$ and $\sigma_E \in  \mathcal{S}_{=}( \mathcal{H}_E)$, we have $
H_2(\rho_{XE}|\sigma_E) \geq H_{\min}(\rho_{XE}|\sigma_E).$
\end{lem}

\begin{lem}[{\cite[Theorem 2]{holenstein2011randomness}}]\label{lem_apphr}
Consider a probability distribution $p_{X^n}\triangleq \prod_{i=1}^n p_{X_i}$	over $\mathcal{X}^n$. For any $\delta \in [0, \log |\mathcal{X}|]$, we have 
\begin{align*}
	\mathbb{P}[- \log p_{X^n} (X^n) \leq H(X^n) -n\delta] &\leq 2^{-\frac{n \delta^2}{2\log(|\mathcal{X}|+3)^2}},\\
	\mathbb{P}[- \log p_{X^n} (X^n) \geq H(X^n) +n\delta] &\leq 2^{-\frac{n \delta^2}{2\log(|\mathcal{X}|+3)^2}}.
\end{align*}
\end{lem}

\begin{lem}[{\cite[Lemma 3.1.10]{renner2008security}}] \label{lem_app4}
For any $\rho_{AB} \in \mathcal{P}( \mathcal{H}_A \otimes \mathcal{H}_B)$ and $\sigma_B \in  \mathcal{P}( \mathcal{H}_B)$,  the fully mixed state on the support of $\rho_B$, we have $
H_{\min}(\rho_{AB}) = H_{\min}(\rho_{AB}|\sigma_B) + H_{\max}(\rho_B).$
\end{lem}

\begin{lem}[{\cite[Lemma 4]{renes2011noisy}}]
\label{lempreimage}
Consider a function $f:\mathcal{X}\rightarrow \mathcal{Y}$ and $\epsilon>0$. We have $ \mathbb{P}[|f^{-1}(Y)|\geq \epsilon |\mathcal{X}|/|\mathcal{Y}|] \geq 1- \epsilon$, where the probability is taken over $Y$ uniformly distributed in $\mathcal{Y}$.
\end{lem}

We next review some definitions and results related to submodular functions.

\begin{defn}[\!\!\cite{edmonds2003submodular,tse1998multiaccess}] \label{defpoly} 
Let $f: 2^{\mathcal{L}} \to \mathbb{R}$. $
\mathcal{P} (f) \triangleq \left\{ (R_{l})_{l \in \mathcal{L}} \in \mathbb{R}_+^L :  R_{\mathcal{S}} \leq f(\mathcal{S}) , \forall \mathcal{S} \subset \mathcal{L} \right\}
$
associated with the function $f$, is a polymatroid if 
\begin{enumerate} [(i)]
\item $f$  is normalized, i.e., $f(\emptyset) =0$,
\item $f$  is non-decreasing, i.e., $\forall \mathcal{S},\mathcal{T} \subset \mathcal{L}, \mathcal{S}\subset\mathcal{T} \implies f(\mathcal{S}) \leq f(\mathcal{T})$,
\item $f$  is submodular, i.e., $\forall \mathcal{S},\mathcal{T} \subset \mathcal{L}, f(\mathcal{S}\cup \mathcal{T}) +  f(\mathcal{S}\cap \mathcal{T}) \leq f(\mathcal{S})+ f(\mathcal{T})$.
\end{enumerate}
\end{defn}

\begin{lem} \label{submodular2}
Let $\rho_{X_{\mathcal{L}}BE}$ be as defined in Theorem \ref{th1}.
\begin{enumerate}[(i)]
\item The set function $h_{\rho}:2^{\mathcal{L}} \to \mathbb{R},\mathcal{S} \mapsto  H(X_{\mathcal{S}}|E )_{\rho}$ is submodular.
\item The set function $g_{\rho}:2^{\mathcal{L}} \to \mathbb{R},\mathcal{S} \mapsto - H(X_{\mathcal{S}}|B X_{\mathcal{S}^c})_{\rho}$ is submodular.
\item The set function $f_{\rho}:2^{\mathcal{L}} \to \mathbb{R},\mathcal{S} \mapsto I(X_{\mathcal{S}};B |X_{\mathcal{S}^c})_{\rho} - I(X_{\mathcal{S}};E)_{\rho}$ is submodular.
\end{enumerate}
 
\end{lem}
\begin{proof}
We first prove $(i)$. For $\mathcal{S},\mathcal{T} \subseteq \mathcal{L}$, we have 
\begin{align*}
& h_{\rho}(\mathcal{S}\cup \mathcal{T}) +  h_{\rho}(\mathcal{S}\cap \mathcal{T}) \\
&   =    H(X_{\mathcal{S} \cup \mathcal{T}}|E)_{\rho}  + H(X_{\mathcal{S} \cap \mathcal{T}}|E)_{\rho}\\
&   =   H(X_{\mathcal{S}}|E)_{\rho} + H(X_{\mathcal{T} \backslash \mathcal{S}}|X_{\mathcal{S}}E)_{\rho}  + H(X_{\mathcal{S} \cap \mathcal{T}}|E)_{\rho}\\
&    \leq  H(X_{\mathcal{S}}|E)_{\rho} + H(X_{\mathcal{T} \backslash \mathcal{S}}|X_{\mathcal{S} \cap \mathcal{T}}E)_{\rho} + H(X_{\mathcal{S} \cap \mathcal{T}}|E)_{\rho}\\
& = h_{\rho}(\mathcal{S})+h_{\rho}( \mathcal{T}),
\end{align*}
where the inequality holds because conditioning does not increase entropy. 

Next, we prove $(ii)$. Remark that for any $\mathcal{S} \subseteq \mathcal{L}$, we have $g_{\rho}(\mathcal{S}) = - H(X_{\mathcal{S}}|B X_{\mathcal{S}^c})_{\rho} =  H(B X_{\mathcal{S}^c})_{\rho} - H(X_{\mathcal{L}}B )_{\rho} = H(X_{\mathcal{S}^c}|B)_{\rho} - H(X_{\mathcal{L}}|B )_{\rho}$, and $\mathcal{S} \mapsto H(X_{\mathcal{S}^c}|B)_{\rho}$ is submodular by $(i)$ since  $\mathcal{S} \mapsto f(\mathcal{S})$ submodular  implies $\mathcal{S} \mapsto f(\mathcal{S}^c)$ submodular. Hence, $g_{\rho}$ is submodular.

Finally, we prove $(iii)$. Remark that   we have $f_{\rho} = g_{\rho}+h_{\rho}$. Hence, since the sum of two submodular functions is submodular, $f_{\rho}$ is submodular.\end{proof}

\begin{lem}[{\cite[Lemma 2]{zhang2017multi}}] \label{lemsubm}
Consider two submodular functions $f:2^{\mathcal{L}} \to \mathbb{R}$ and $g:2^{\mathcal{L}} \to \mathbb{R}$. Then, the following system of equations for $(x_l)_{l\in\mathcal{L}} \in \mathbb{R}_+^L$
\begin{align*}
-g(\mathcal{S})\leq\sum_{s \in \mathcal{S}} x_s \leq f(\mathcal{S}) ,\forall \mathcal{S} \subseteq \mathcal{L},
\end{align*}
has a solution if and only if $-g(\mathcal{S}) \leq f(\mathcal{S}) ,\forall \mathcal{S} \subseteq \mathcal{L}$.

\end{lem}

\begin{lem}[{\cite[Lemma 9]{chou2018polar}}] \label{propolyplus}
Let $f: 2^{\mathcal{L}} \to \mathbb{R}$ be  a positive, normalized, and submodular function. Then,  
$$
f^{*}:2^{\mathcal{L}} \to \mathbb{R}_+, \mathcal{S} \mapsto   \min_{ \substack{\mathcal{A} \subseteq \mathcal{L} \\ \text{s.t. } \mathcal{A} \supseteq \mathcal{S}}} f (\mathcal{A}).
$$
is normalized, non-decreasing, and submodular. 

\end{lem}

\bibliographystyle{IEEEtran}
\bibliography{polarwiretap}

\end{document}